\documentclass[a4paper,onecolumn,11pt,unpublished]{quantumarticle}
\pdfoutput=1

\usepackage{amsmath,amssymb,amsthm,mathtools,bm,physics}
\usepackage{graphicx}
\usepackage{comment}
\usepackage[normalem]{ulem}

\usepackage[numbers,sort&compress]{natbib}

\usepackage{hyperref}

\newcommand{\Com}[2]{[#1,#2]}
\newcommand{\CA}{\mathbb{C}}

\newcommand{\lM}{\lambda_{\rm max}}
\newcommand{\lm}{\lambda_{\rm min}}

\newcommand{\id}{I}

\newcommand{\I}{\mathbb{I}}
\newcommand{\rnorm}[1]{\lvert #1\rvert}

\theoremstyle{plain}
\newtheorem{theorem}{Theorem}
\newtheorem{lemma}{Lemma}
\newtheorem{proposition}{Proposition}

\newtheorem{remark}{Remark}

\begin{document}
\title{Optimal Uncertainty Relations for a Single Observable}

\author{Haruki Yamashita}%
\orcid{0009-0002-2439-0317}
\email{haruki.yamashita1367@gmail.com}
\affiliation{
Graduate School of Engineering and Science, Shibaura Institute of Technology,
307~Fukasaku, Minuma-ku, Saitama 337-0003, Japan
}

\author{Aina Mayumi}%
\orcid{0009-0007-0956-1620}
\affiliation{
Graduate School of Engineering and Science, Shibaura Institute of Technology,
307~Fukasaku, Minuma-ku, Saitama 337-0003, Japan
}

\author{Gen Kimura}%
\orcid{0000-0003-4288-2024}
\email{gen.kimura.quant@gmail.com}
\thanks{Corresponding author}
\affiliation{
Graduate School of Engineering and Science, Shibaura Institute of Technology,
307~Fukasaku, Minuma-ku, Saitama 337-0003, Japan
}
\affiliation{
Graduate School of Information Sciences, Tohoku University, Sendai 980-8579, Japan
}

\date{}
\begin{abstract}
Uncertainty relations are usually formulated as trade-offs between two or more observables. Yet the uncertainty of a single observable can already contain an intrinsically quantum contribution arising from its noncommutativity with the state. In the sense of Luo, such a contribution can be quantified by a quantum uncertainty $Q_\rho(A)$, with prominent examples including the Wigner--Yanase and Wigner--Yanase--Dyson skew informations and the quantum Fisher information. Here we ask how strong such single-observable uncertainty relations can be when the quantum state is fixed.

We solve this fixed-state optimization problem by determining the largest coefficient $c_{\rm opt}(\rho)$ for which $V_\rho(A)\ge c_{\rm opt}(\rho)Q_\rho(A)$ holds for every observable $A$, and show that it depends only on the smallest and largest eigenvalues of $\rho$. The same optimization principle extends to a new power-commutator family $\frac{1}{2}|[\rho^s,A]|^2$ $(s\ge 1/2)$, providing a direct measure of state--observable noncommutativity beyond the standard quantum-uncertainty framework.

More importantly, the optimal noncommutative contribution can be combined with the maximal classical contribution compatible with the state, yielding a sharpened variance decomposition without weakening the optimal bound. For states with exactly two distinct eigenvalues, this relation becomes an identity for every observable; in particular, qubit variances decompose exactly into classical and noncommutative parts. Finally, the corresponding variance-product bounds are no weaker than the Robertson relation for qubit spin observables and can be strictly stronger, showing that a substantial part of the uncertainty usually expressed as a trade-off between distinct observables is already encoded in their individual noncommutativity with the state.
\end{abstract}

\maketitle
\section{Introduction}

The uncertainty principle is one of the fundamental features of quantum mechanics. In its standard preparation-uncertainty formulation, it constrains the simultaneous statistical spreads of two or more observables in a given quantum state. 
Starting from the position--momentum relation of Kennard~\cite{Kennard1927}, this viewpoint was extended by Robertson to arbitrary pairs of observables~\cite{Robertson1929} and subsequently sharpened by Schr\"odinger~\cite{Schrodinger1930}. 
Numerous further developments have produced
variance-based~\cite{
Luo2,Park2005,Yichen,Maccone,Li2015,ChenFei2015SumUR,
SONG20162925,Xiao2016,ZhengZhang2017,Fan2020,Zhang2021,
math4010008,MKOC_2024,KMOLC,KMY},
entropic~\cite{
IBB1,Deutsch,Kraus_1987,Maassen,Rastegin,Berta_2010,HallRenyi},
information-theoretic~\cite{
Luo2003WignerYanaseSkewInformation,LZ2004,Yanagi_2005,
Kosaki_2005,Li_2009,FURUICHI2009179,Yanagi_2010,
Gibilisco_2007,GibiliscoHiaiPetz2009,CG2022,TGFF2022},
and coherence-based uncertainty relations~\cite{
SPB2016,Liu2016,PhysRevA.96.032313,Rastegin2017,
TajimaNagaoka2019CoherenceVariance}.
For broader reviews of uncertainty relations and their developments,
see Refs.~\cite{BUSCH2007155,Wehner_2010,RevModPhys.89.015002,HilgevoordUffink2024}.

These formulations predominantly express uncertainty through a trade-off between different observables. This raises a more elementary question: \emph{is there a fundamental lower bound of quantum origin on the uncertainty of a single observable?}

A simple observation suggests that such a bound can indeed arise from the noncommutativity between the observable and the quantum state. 
Let $A$ be an observable and $\rho$ a quantum state.
Denoting the expectation value by $a=\Tr(\rho A)$, the variance can be expressed as
\begin{equation}
V_\rho(A)=\Tr\!\left[\rho(A-aI)^2\right]=\|(A-aI)\sqrt{\rho}\|^2.
\end{equation}
Here and throughout this paper, $\|\cdot\|$ denotes the Hilbert--Schmidt norm, $\|X\|:=\sqrt{\Tr(X^\dagger X)}$.
By the definiteness of the Hilbert--Schmidt norm, $V_\rho(A)=0$ implies $(A-aI)\sqrt{\rho}=0$, and hence $A\rho=a\rho$, and therefore $\rho$ and $A$ commute. 
Consequently, whenever the state and the observable do not commute, the uncertainty of that observable cannot vanish.
This simple fact suggests that state--observable noncommutativity can give rise to a quantum-mechanical lower bound on the uncertainty of a single observable.

Quantitative manifestations of this idea have appeared in various forms through quantities such as skew information and quantum Fisher information.
Following Luo~\cite{LuoQuantumClassical}, for a state $\rho$ and an observable $A$, a nonnegative quantity $Q_\rho(A)$ may be regarded as a \emph{quantum uncertainty} if it satisfies: (i) convexity in $\rho$, (ii) $Q_\rho(A)=V_\rho(A)$ for every pure state $\rho$, and (iii) $Q_\rho(A)=0$ whenever $[A,\rho]=0$.
Then, for any pure-state decomposition $\rho=\sum_k p_k|\psi_k\rangle\langle\psi_k|$, conditions (i), (ii) and concavity of $V_\rho(A)$ give
\begin{equation}
Q_\rho(A)\le\sum_k p_k Q_{\psi_k}(A)=\sum_k p_k V_{\psi_k}(A)\le V_\rho(A).
\end{equation}
Thus one obtains a quantitative uncertainty relation for a single observable,
\begin{equation}
V_\rho(A)\ge Q_\rho(A).
\label{eq:general-single}
\end{equation}
Condition (iii) further shows that the lower bound $Q_\rho(A)$ is associated with the noncommutativity between the observable and the state.

Typical examples of quantum uncertainty are provided by the Wigner--Yanase skew information, the Wigner--Yanase--Dyson skew information, and the quantum Fisher information.
The Wigner--Yanase (WY) skew information~\cite{wigner1963} is defined by
\begin{equation}
I_\rho(A):=-\frac12\Tr[\sqrt{\rho},A]^2=\frac12\|[\sqrt{\rho},A]\|^2.\label{eq:WY}
\end{equation}
A natural extension is the Wigner--Yanase--Dyson (WYD) skew information,
\begin{equation}
I_{\rho,\alpha}(A):=-\frac12\Tr\!\left([\rho^\alpha,A][\rho^{1-\alpha},A]\right),\qquad 0<\alpha<1.\label{eq:WYD}
\end{equation}
Its convexity in $\rho$ follows from Lieb's celebrated concavity theorem~\cite{Lieb1973ConvexTraceFunctions}.
Another important example is provided by one quarter of the symmetric-logarithmic-derivative quantum Fisher information (QFI)~\cite{BraunsteinCaves1994},
\begin{equation}
\tilde{F}_{\rho}(A):= \frac{1}{4}F_{Q}[\rho,A]
:=
\frac{1}{2}\sum_{\substack{i,j\\ \lambda_i+\lambda_j>0}}
\frac{(\lambda_i-\lambda_j)^2}{\lambda_i+\lambda_j}
|\langle i|A|j\rangle|^2,
\label{eq:QFI-def}
\end{equation}
where $\rho=\sum_i\lambda_i|i\rangle\langle i|$ is a spectral decomposition of $\rho$.
All of these quantities are special cases of Hansen's metric-adjusted skew information~\cite{Hansen2008MetricAdjustedSkewInformation}.

Outside Hansen's metric-adjusted class, one may consider a particularly simple quantity that captures the noncommutativity between the state and the observable more directly:
\begin{equation}
K_{\rho}(A):=\frac{1}{2}\|[A,\rho]\|^2.\label{eq:K}
\end{equation}
A natural one-parameter extension is
\begin{equation}
K_{\rho,s}(A):=\frac{1}{2}\|[A,\rho^s]\|^2,\qquad \frac{1}{2}\le s.\label{eq:Ks}
\end{equation}
The endpoint $s=1/2$ coincides with the Wigner--Yanase skew information, whereas $s=1$ gives the direct commutator expression above.
For $1/2\le s\le1$, the quantity $K_{\rho,s}(A)$ qualifies as a quantum uncertainty in the sense introduced above.
In particular, its convexity in $\rho$ for $1/2<s<1$ is highly nontrivial and is established in Ref.~\cite{gen2026}, while the case $s=1$ follows directly from the linearity of $[A,\rho]$ in $\rho$.
For $s>1$, by contrast, $K_{\rho,s}(A)$ is not convex in general and therefore does not define a quantum uncertainty in this sense.
Nevertheless, we define $K_{\rho,s}(A)$ for the entire range $s\ge1/2$, because, as shown below, the corresponding single-observable uncertainty relation remains valid throughout this wider range.

Notably, it is known that, among all quantities $Q_\rho(A)$ that are convex in $\rho$ and coincide with the variance for pure states, $\frac14F_Q[\rho,A]$ is the largest~\cite{TothPetz2013Extremal,Yu2013QFIConvexRoof}.
Thus, within this class, it provides the strongest universal single-observable lower bound.

\paragraph{Main results and contributions.}

The main problem addressed in this work is the fixed-state optimization of single-observable uncertainty relations.
Specifically, for a given quantum uncertainty $Q_\rho(A)$, we first consider a state-dependent refinement of Eq.~\eqref{eq:general-single} of the form
\begin{equation}
V_\rho(A)\ge c(\rho)Q_\rho(A),\qquad c(\rho)\ge1,
\label{eq:GeneralUR}
\end{equation}
and, for each fixed (nonmaximally mixed) state $\rho$, seek the largest coefficient $c_{\rm opt}(\rho)$ for which this inequality holds for every observable $A$.

For the families considered above, we solve this optimization problem explicitly for any finite-dimensional quantum system on $\CA^d$. 
We determine the fixed-state optimal coefficients for the Wigner--Yanase skew information $I_\rho(A)$, the full Wigner--Yanase--Dyson family $I_{\rho,\alpha}(A)$, the quantum Fisher information $\tilde{F}_{\rho}(A)$, and the power-commutator family $K_{\rho,s}(A)$, which includes $K_\rho(A)$ as the case $s=1$.

Let $\lambda_{\min}$ and $\lambda_{\max}$ denote the smallest and largest eigenvalues of $\rho$, respectively.
Notice that, when $\rho=\I/d$ is maximally mixed, i.e., $\lambda_{\min}=\lambda_{\max}$, all the corresponding quantum uncertainties vanish identically for every observable.
Accordingly, the lower bounds in the relations displayed below are to be understood as zero, despite the apparent singularities in the coefficients.
For the WY skew information \eqref{eq:WY} and the direct state--observable commutator \eqref{eq:K}, we obtain
\begin{align}
V_\rho(A)
&\ge
\frac{\lambda_{\max}+\lambda_{\min}}
{\bigl(\sqrt{\lambda_{\max}}-\sqrt{\lambda_{\min}}\bigr)^2}I_\rho(A),
\label{eq:WY-opt-intro}
\end{align}
and
\begin{align}
V_\rho(A)
&\ge
\frac{\lambda_{\max}+\lambda_{\min}}
{(\lambda_{\max}-\lambda_{\min})^2}
K_\rho(A).
\label{eq:rho-opt-intro}
\end{align}
More generally, for the WYD family \eqref{eq:WYD}, the optimal relation is
\begin{equation}
V_\rho(A)
\ge
\frac{\lambda_{\max}+\lambda_{\min}}
{(\lambda_{\max}^{\alpha}-\lambda_{\min}^{\alpha})
(\lambda_{\max}^{1-\alpha}-\lambda_{\min}^{1-\alpha})}
I_{\rho,\alpha}(A), \qquad 0<\alpha<1
\label{eq:WYD-sharp-intro}
\end{equation}
while for the power-commutator family \eqref{eq:Ks} we obtain
\begin{equation}
V_\rho(A)
\ge
\frac{\lambda_{\max}+\lambda_{\min}}
{(\lambda_{\max}^s-\lambda_{\min}^s)^2}
K_{\rho,s}(A),
\qquad \frac12\le s
\end{equation}
The cases $\alpha=1/2$ and $s=1/2$ both reproduce relation \eqref{eq:WY-opt-intro}, while $s=1$ reproduces relation \eqref{eq:rho-opt-intro}.
For the quantum Fisher information \eqref{eq:QFI-def}, we obtain
\begin{equation}
V_\rho(A)
\ge
\left(
\frac{\lambda_{\max}+\lambda_{\min}}
{\lambda_{\max}-\lambda_{\min}}
\right)^2
\tilde{F}_\rho(A).
\label{eq:QFI-sharp-intro}
\end{equation}
This result is particularly notable because $\tilde{F}_\rho(A)$ already provides the strongest universal single-observable lower bound within the convex class described above, yet it admits a strict state-dependent improvement for every full-rank nonmaximally mixed state.

In this paper, we further show that these optimal multiplicative bounds can be strengthened by incorporating a maximal classical contribution to the variance.
To this end, we introduce in Sec.~\ref{sec:main} a classical uncertainty $V_\rho^{\rm cl}(A)$, defined in Eq.~\eqref{eq:Vcl}, which captures the maximal contribution to the variance compatible with the state.
We then establish the general uncertainty relation
\begin{equation}
V_\rho(A)\ge V_\rho^{\rm cl}(A)+c_{\rm opt}(\rho)Q_\rho(A),
\label{eq:GeneralUR2}
\end{equation}
for each of the WY, WYD, QFI, and power-commutator families considered above. 
Importantly, the coefficient $c_{\rm opt}(\rho)$ remains optimal even after the classical contribution is included.
Thus, the optimal quantum lower bound can be supplemented by the maximal classical contribution without any loss in the strength of its noncommutative part.
We refer to relations of these forms as
\emph{single-observable preparation uncertainty relations}.

\begin{remark} 
Several remarks are in order.
First, Eq.~\eqref{eq:GeneralUR2} provides a general uncertainty relation for an arbitrary state $\rho$ and any single observable $A$.
In particular, when $\rho$ is maximally mixed, the quantum-uncertainty term vanishes, while the classical contribution coincides with the total variance, $V_\rho^{\rm cl}(A)=V_\rho(A)$.
Thus, Eq.~\eqref{eq:GeneralUR2} is saturated in this case.

Second, for every faithful nonmaximally mixed state, all of the optimal coefficients above are strictly greater than $1$, yielding genuine state-dependent improvements over the original bounds.
For a non-faithful state, i.e., $\lambda_{\min}=0$, the optimal coefficients for the WY, WYD, and QFI bounds reduce to $1$, so these relations coincide with the original bound \eqref{eq:general-single}.
By contrast, for the power-commutator family one has
\begin{equation}
c_s^{\rm opt}(\rho)=\lambda_{\max}^{1-2s},
\end{equation}
which is strictly greater than $1$ for every non-pure state when $s>1/2$, since then $0<\lambda_{\max}<1$.
For a pure state, $\lambda_{\max}=1$, and hence $c_s^{\rm opt}(\rho)=1$.

Finally, and most importantly, in every case the displayed coefficient is optimal for the fixed state $\rho$: it is the largest possible coefficient for which the corresponding inequality remains valid for every observable $A$.
Thus, these relations are not merely state-dependent refinements of the conventional bounds, but the strongest possible bounds of the form \eqref{eq:GeneralUR} for the given state.
Consequently, the fact that the optimal coefficient depends only on the two extremal eigenvalues $\lambda_{\min}$ and $\lambda_{\max}$ should not be regarded as an artifact of our derivation, but rather as an intrinsic feature of these uncertainty relations.
Interestingly, the same dependence on the smallest and largest eigenvalues also appears in the optimization of Robertson-type uncertainty relations~\cite{KMY}.
\end{remark}

We next turn to a complementary question concerning equality: when do these sharpened relations become exact identities for every observable?
We show that, within the class of full-rank nonmaximally mixed states, these sharpened relations become exact identities for every observable if and only if the state has exactly two distinct eigenvalues, for all quantum uncertainties considered above.
Thus, universal exactness is completely characterized by the number of distinct eigenvalues of the state.

The origin of this exactness can be understood directly from the spectral structure of the state.
If $\rho$ has exactly two distinct eigenvalues, every noncommuting contribution of an observable connects the same pair of eigenspaces.
Consequently, both the variance beyond its classical part and the noncommutative quantities appearing in the bounds are determined by the same off-diagonal block.
By contrast, if three or more distinct eigenvalues are present, one can choose an observable supported on a pair of eigenspaces other than the pair corresponding to $\lambda_{\min}$ and $\lambda_{\max}$, for which the scalar inequality underlying the optimal bound is strict.
Thus, the presence of a third distinct eigenvalue prevents universal saturation.

Qubit systems provide the simplest realization of this structure.
Every full-rank nonmaximally mixed qubit state has exactly two distinct eigenvalues, and hence all of the sharpened relations above become exact identities for every qubit observable.
The pure and maximally mixed boundary cases can be treated separately, and the identities therefore extend to every qubit state.
For qubits, the variance of a single observable thus admits an exact decomposition into a classical contribution and a noncommutative contribution.
The Bloch-sphere representation further provides a simple geometric picture of the unsharpened relations and their equality conditions.

Finally, we return to the conventional setting of uncertainty relations involving several observables.
Since the single-observable relations derived here can be applied separately to each observable, they can be combined to produce conventional product-form uncertainty relations.
Because these bounds do not explicitly involve commutators between distinct observables, one might expect the resulting relations to be weaker than conventional trade-off bounds.
The qubit case, however, shows that this expectation need not hold.

For qubits, the sharpened single-observable relations are exact, and their products therefore yield exact identities for products of variances.
Even after the classical contributions are discarded, the resulting optimized noncommutative bounds remain unexpectedly strong.
In particular, for every pair of qubit spin observables, the resulting product bound is no weaker than the Robertson bound.
Moreover, for mixed qubit states with $1/2<P<1$, where $P=\Tr\rho^2$ is the purity, its uniform average exceeds the corresponding averaged Schr\"odinger and Luo-type bounds.
These results show that a substantial part of the uncertainty conventionally formulated as a trade-off between different observables is already encoded in the constraints imposed independently on each observable by its relation to the state.

The paper is organized as follows.
In Sec.~2, we derive the optimal single-observable preparation uncertainty relations, introduce $V_\rho^{\rm cl}(A)$, and establish the sharpened power-commutator, Wigner--Yanase--Dyson, and quantum-Fisher-information relations.
In Sec.~3, we establish the necessary-and-sufficient characterization of universal exactness in terms of the number of distinct eigenvalues and then specialize the results to qubit systems.
In Sec.~4, we apply the single-observable relations to variance-product uncertainty relations and compare the resulting bounds for qubits.
Sec.~5 summarizes the results and discusses possible extensions.

\section{Single-observable preparation uncertainty relations}\label{sec:main}

In this section, we establish the sharpened single-observable uncertainty relations of the form \eqref{eq:GeneralUR2}.
We first introduce a notion of classical uncertainty, denoted by $V_\rho^{\rm cl}(A)$, which constitutes another conceptual ingredient of our approach.
The basic idea is to identify, within the total variance $V_\rho(A)$, the maximal contribution that can be attributed to the part of the observable compatible with the state $\rho$.
If $\rho$ has degenerate eigenvalues, its eigenbasis is not unique.
We therefore define the classical contribution by maximizing over all eigenbases of $\rho$, so that the resulting quantity does not depend on an arbitrary choice of basis.
We then show that this optimized quantity has a simple basis-independent representation in terms of the pinching map associated with $\rho$.
This representation naturally separates the variance into a classical contribution and a remaining contribution arising from the noncommutativity between $A$ and $\rho$.
Using this decomposition, we establish the optimal sharpened relations for the power-commutator family $K_{\rho,s}(A)$ and the Wigner--Yanase--Dyson skew information $I_{\rho,\alpha}(A)$, followed by the corresponding result for the quantum Fisher information.
The Wigner--Yanase and direct commutator relations \eqref{eq:WY-opt-intro} and \eqref{eq:rho-opt-intro} are included as the special cases $s=1/2$ and $s=1$, respectively.

\subsection{Classical uncertainty and variance decomposition}

Let $\{|\phi_i\rangle\}_{i=1}^d$ be an orthonormal eigenbasis of $\rho$.
We define the part of an observable $A$ that is diagonal in this basis by
\begin{equation}\label{eq:clpartA}
A^{\rm cl}_{\{\phi_i\}}
:=
\sum_{i=1}^d
\langle\phi_i|A|\phi_i\rangle
|\phi_i\rangle\langle\phi_i|.
\end{equation}
A classical system may be embedded into a quantum system by restricting states and observables to those diagonal in a common basis.
It is therefore natural to regard $A^{\rm cl}_{\{\phi_i\}}$ as the classical part of $A$ relative to the chosen eigenbasis of $\rho$.

If $\rho$ has degenerate eigenvalues, however, its eigenbasis is not unique.
We remove this ambiguity by maximizing over all eigenbases of $\rho$ and define the classical uncertainty by
\begin{equation}\label{eq:Vcl}
V_\rho^{\rm cl}(A)
:=
\sup_{\{|\phi_i\rangle\}\,:\,\text{eigenbasis of }\rho}
V_\rho\!\left(A^{\rm cl}_{\{\phi_i\}}\right).
\end{equation}
Thus, $V_\rho^{\rm cl}(A)$ represents the maximal classical contribution to the variance of $A$ compatible with the state $\rho$.

This optimization admits a simple basis-independent representation.
Let
\begin{equation}
\rho=\sum_\mu\lambda_\mu P_\mu
\end{equation}
be the spectral decomposition of $\rho$, where the $\lambda_\mu$ are the distinct eigenvalues.
We denote by
\begin{equation}
\mathcal P_\rho(A):=\sum_\mu P_\mu A P_\mu
\end{equation}
the block-diagonal part of $A$ with respect to the eigenspace decomposition of $\rho$.
This is the pinching of $A$ associated with $\rho$ \cite{Hayashi2006QI}.
\begin{proposition}\label{prop:classical-pinching}
For any state $\rho$ and for every observable $A$,
\begin{equation}
V_\rho^{\rm cl}(A)=V_\rho(\mathcal P_\rho(A)) = V_\rho \left(\sum_\mu P_\mu A P_\mu\right).
\end{equation}
\end{proposition}

\begin{proof}
Fix an arbitrary eigenbasis $\{|\phi_{\mu k}\rangle\}_{\mu,k}$ of $\rho$, where $|\phi_{\mu k}\rangle\in\operatorname{Ran}P_\mu$, and set $A_\mu:=P_\mu A P_\mu$. 
Since $A_\mu$ is self-adjoint, 
\begin{equation}
\sum_k|\langle\phi_{\mu k}|A_\mu|\phi_{\mu k}\rangle|^2\le \sum_{k,l} |\langle\phi_{\mu k}|A_\mu|\phi_{\mu l}\rangle|^2 = \Tr(A_\mu^2).
\end{equation}
Moreover, since $A^{\rm cl}_{\{\phi_{\mu k}\}}$ has the same diagonal matrix elements as $A$ in an eigenbasis of $\rho$,
\begin{equation}
\Tr\left(\rho A^{\rm cl}_{\{\phi_{\mu k}\}}\right)=\Tr(\rho A)=\Tr\left(\rho\mathcal P_\rho(A)\right).
\end{equation}
For the second moments, using $\rho P_\mu=\lambda_\mu P_\mu$, we have
\begin{align}
\Tr\left(\rho\left(A^{\rm cl}_{\{\phi_{\mu k}\}}\right)^2\right)=\sum_\mu\lambda_\mu\sum_k|\langle\phi_{\mu k}|A_\mu|\phi_{\mu k}\rangle|^2 \le\sum_\mu\lambda_\mu\Tr(A_\mu^2) =\Tr\left(\rho\mathcal P_\rho(A)^2\right).
\end{align}
Since the two observables have the same expectation value with respect to $\rho$, subtracting the common squared mean yields
\begin{equation}
V_\rho\left(A^{\rm cl}_{\{\phi_{\mu k}\}}\right)\le V_\rho\left(\mathcal P_\rho(A)\right).
\end{equation}
Taking the supremum over all eigenbases gives
\begin{equation}
V_\rho^{\rm cl}(A)\le V_\rho\!\left(\mathcal P_\rho(A)\right).
\end{equation}
Conversely, within each eigenspace $\operatorname{Ran}P_\mu$, choose an orthonormal basis diagonalizing $A_\mu$.
For the resulting eigenbasis of $\rho$, one has
\begin{equation}
A^{\rm cl}_{\{\phi_{\mu k}\}}=\sum_\mu A_\mu=\mathcal P_\rho(A),
\end{equation}
so the upper bound is attained.
\end{proof}

Proposition~\ref{prop:classical-pinching} shows that the classical uncertainty defined through the optimization in Eq.~\eqref{eq:Vcl} admits a basis-independent characterization.
For the proofs of the uncertainty relations below, however, it is useful to return to the original definition and fix an arbitrary eigenbasis of $\rho$.
Let
\begin{equation}
\rho=\sum_i\lambda_i|\phi_i\rangle\langle\phi_i|
\end{equation}
be such a spectral decomposition, and set $A_{ij}:=\langle\phi_i|A|\phi_j\rangle$ and $\mu:=\Tr(\rho A)$.
Then
\begin{align}
V_\rho(A)
&=\Tr(\rho A^2)-\mu^2 \nonumber\\
&=\sum_i\lambda_i A_{ii}^2-\mu^2+\sum_{i<j}(\lambda_i+\lambda_j)|A_{ij}|^2 \nonumber\\
&=V_\rho\left(A^{\rm cl}_{\{\phi_i\}}\right)+\sum_{i<j}(\lambda_i+\lambda_j)|A_{ij}|^2.
\label{eq:variance-basis-decomp}
\end{align}
This decomposition holds for every eigenbasis of $\rho$ and will be the common starting point for the uncertainty relations derived below.

\begin{remark} The notion of classical uncertainty introduced above should be distinguished from several related quantum--classical decompositions of uncertainty in the literature. In particular, Luo defined a classical contribution to the variance as the part complementary to a chosen quantum uncertainty such as the Wigner--Yanase skew information~\cite{LuoQuantumClassical}. 
Our construction takes a different route: $V_\rho^{\rm cl}(A)$ is defined directly from the maximal part of the observable that is compatible with the state $\rho$, independently of any particular choice of quantum uncertainty $Q_\rho(A)$. 

Related ideas also appear in Hall's work on classical and nonclassical components of observables, where the classical component is associated with optimal estimation of one observable from another~\cite{HallExact}, and more recently in general quantum--classical decompositions of uncertainty and other resources~\cite{HallRenyi}. 
The present construction is instead tied specifically to the spectral structure of the prepared state: the compatible part of $A$ is given by the pinching $\mathcal P_\rho(A)$, and its variance yields $V_\rho^{\rm cl}(A)$. 
Thus, while these approaches are conceptually related, the quantity introduced here is defined directly through state--observable compatibility and is tailored to the single-observable uncertainty relations studied in this work. 
\end{remark}

\subsection{Power-commutator, Wigner--Yanase--Dyson, and Quantum-Fisher-information relations}

We first consider the power-commutator family $K_{\rho,s}(A)$ defined in Eq.~\eqref{eq:Ks}.
The following scalar inequality is the key ingredient; its proof is given in Appendix~\ref{app:lemma}.

\begin{lemma}\label{lem:sq}
Let $0\le m<M$ and $s\ge1/2$.
Then, for any $x,y\in[m,M]$,
\begin{equation}
x+y
\ge
\frac{M+m}{(M^s-m^s)^2}(x^s-y^s)^2.
\label{eq:scalar-ineq-s}
\end{equation}
\end{lemma}

We can now state the optimal sharpened relation.

\begin{theorem}\label{thm:s-main}
For every density operator $\rho$, every observable $A$, and $s\ge1/2$,
\begin{equation}
V_\rho(A)
\ge
V_\rho^{\rm cl}(A)
+
\frac{\lM+\lm}{(\lM^s-\lm^s)^2}
K_{\rho,s}(A).
\label{eq:mainUR}
\end{equation}
When $\rho$ is maximally mixed, the second term on the right-hand side is understood to be zero.
For every nonmaximally mixed state $\rho$, the relation is tight.
In particular, the coefficient
\begin{equation}\label{eq:optCoef}
c_s^{\rm opt}(\rho)
=
\frac{\lM+\lm}{(\lM^s-\lm^s)^2}
\end{equation}
is optimal: it cannot be replaced by any larger coefficient while preserving Eq.~\eqref{eq:mainUR} for every observable $A$.
\end{theorem}
\begin{remark} For $1/2\le s\le1$, the quantity $K_{\rho,s}(A)$ is convex in $\rho$ and hence qualifies as a quantum uncertainty in the sense discussed above. 
For $s>1$, however, $K_{\rho,s}(A)$ is not convex in general and therefore does not define a quantum uncertainty; see Ref.~\cite{gen2026}. Nevertheless, it is worth noting that the uncertainty relation~\eqref{eq:mainUR} itself remains valid for the entire range $s\ge1/2$. \end{remark}

\begin{proof}
Assume first that $\lm<\lM$.
Fix an arbitrary orthonormal eigenbasis $\{|\phi_i\rangle\}_{i=1}^d$ of $\rho$, write
\begin{equation}
\rho=\sum_i\lambda_i|\phi_i\rangle\langle\phi_i|,
\end{equation}
and set $A_{ij}:=\langle\phi_i|A|\phi_j\rangle$.
By the decomposition \eqref{eq:variance-basis-decomp},
\begin{equation}
V_\rho(A)-V_\rho\left(A^{\rm cl}_{\{\phi_i\}}\right)
=
\sum_{i<j}(\lambda_i+\lambda_j)|A_{ij}|^2.
\end{equation}
On the other hand, from the definition of $K_{\rho,s}(A)$,
\begin{equation}
K_{\rho,s}(A)
=
\sum_{i<j}
(\lambda_i^s-\lambda_j^s)^2|A_{ij}|^2.
\end{equation}
Applying Lemma~\ref{lem:sq} with $m=\lm$, $M=\lM$, $x=\lambda_i$, and $y=\lambda_j$, we obtain
\begin{equation}
\lambda_i+\lambda_j
\ge
\frac{\lM+\lm}{(\lM^s-\lm^s)^2}
(\lambda_i^s-\lambda_j^s)^2.
\end{equation}
Multiplying this inequality by $|A_{ij}|^2$ and summing over $i<j$ gives
\begin{equation}
V_\rho(A)
\ge
V_\rho\left(A^{\rm cl}_{\{\phi_i\}}\right)
+
\frac{\lM+\lm}{(\lM^s-\lm^s)^2}
K_{\rho,s}(A).
\end{equation}
Since the eigenbasis $\{|\phi_i\rangle\}$ was arbitrary, taking the supremum over all eigenbases of $\rho$ and using Eq.~\eqref{eq:Vcl} yields Eq.~\eqref{eq:mainUR}.

It remains to establish tightness and optimality.
Let $|\psi_{\min}\rangle$ and $|\psi_{\max}\rangle$ be unit eigenvectors corresponding to $\lm$ and $\lM$, respectively, and define
\begin{equation}\label{eq:equality}
A_0
=
|\psi_{\max}\rangle\langle\psi_{\min}|
+
|\psi_{\min}\rangle\langle\psi_{\max}|.
\end{equation}
For this observable,
\begin{equation}
V_\rho^{\rm cl}(A_0)=0,\qquad
V_\rho(A_0)=\lM+\lm,
\end{equation}
and
\begin{equation}
K_{\rho,s}(A_0)=(\lM^s-\lm^s)^2.
\end{equation}
Hence Eq.~\eqref{eq:mainUR} is saturated by $A_0$.
Any increase of the coefficient in Eq.~\eqref{eq:mainUR} would therefore violate the inequality for this observable, proving optimality.

If $\rho$ is maximally mixed, then $\mathcal P_\rho(A)=A$, and hence $V_\rho^{\rm cl}(A)=V_\rho(A)$, while $K_{\rho,s}(A)=0$ for every observable.
Thus Eq.~\eqref{eq:mainUR} is saturated in this case as well.
\end{proof}

Since
\begin{equation}
K_{\rho,1/2}(A)=I_\rho(A),\qquad K_{\rho,1}(A)=K_\rho(A),
\end{equation}
Theorem~\ref{thm:s-main} simultaneously yields the sharpened Wigner--Yanase and direct-commutator relations.
Dropping the nonnegative term $V_\rho^{\rm cl}(A)$ immediately gives Eqs.~\eqref{eq:WY-opt-intro} and \eqref{eq:rho-opt-intro}.

We next consider the full Wigner--Yanase--Dyson family.
Again, the proof reduces to a scalar inequality, whose proof is given in Appendix~\ref{app:lemma}.
\begin{lemma}\label{lem:wyd-scalar}
Let $0\le m<M$ and $0<\alpha<1$.
Then, for any $x,y\in[m,M]$,
\begin{equation}\label{eq:wyd-scalar}
x+y
\ge
(M+m)
\frac{(x^\alpha-y^\alpha)(x^{1-\alpha}-y^{1-\alpha})}
{(M^\alpha-m^\alpha)(M^{1-\alpha}-m^{1-\alpha})}.
\end{equation}
\end{lemma}

\begin{theorem}\label{thm:wyd-extension}
For every density operator $\rho$, every observable $A$, and $0<\alpha<1$,
\begin{equation}\label{eq:wyd-extension}
V_\rho(A)
\ge
V_\rho^{\rm cl}(A)
+
\frac{\lM+\lm}
{(\lM^\alpha-\lm^\alpha)(\lM^{1-\alpha}-\lm^{1-\alpha})}
I_{\rho,\alpha}(A).
\end{equation}
When $\rho$ is maximally mixed, the second term on the right-hand side is understood to be zero.
For every nonmaximally mixed state $\rho$, the relation is tight.
In particular, the displayed coefficient is optimal: it cannot be replaced by any larger coefficient while preserving Eq.~\eqref{eq:wyd-extension} for every observable $A$.
\end{theorem}

\begin{proof}
Assume first that $\lm<\lM$.
Fix an arbitrary orthonormal eigenbasis $\{|\phi_i\rangle\}_{i=1}^d$ of $\rho$ and set $A_{ij}:=\langle\phi_i|A|\phi_j\rangle$.
By Eq.~\eqref{eq:variance-basis-decomp},
\begin{equation}
V_\rho(A)-V_\rho\left(A^{\rm cl}_{\{\phi_i\}}\right)
=
\sum_{i<j}(\lambda_i+\lambda_j)|A_{ij}|^2.
\end{equation}
From Eq.~\eqref{eq:WYD},
\begin{equation}
I_{\rho,\alpha}(A)
=
\sum_{i<j}
(\lambda_i^\alpha-\lambda_j^\alpha)
(\lambda_i^{1-\alpha}-\lambda_j^{1-\alpha})
|A_{ij}|^2.
\label{eq:wyd-eigen-expansion}
\end{equation}
Applying Lemma~\ref{lem:wyd-scalar} with $m=\lm$, $M=\lM$, $x=\lambda_i$, and $y=\lambda_j$, multiplying by $|A_{ij}|^2$, and summing over $i<j$, we obtain
\begin{equation}
V_\rho(A)
\ge
V_\rho\left(A^{\rm cl}_{\{\phi_i\}}\right)
+
\frac{\lM+\lm}
{(\lM^\alpha-\lm^\alpha)(\lM^{1-\alpha}-\lm^{1-\alpha})}
I_{\rho,\alpha}(A).
\end{equation}
Since the eigenbasis was arbitrary, taking the supremum over all eigenbases of $\rho$ and using Eq.~\eqref{eq:Vcl} proves Eq.~\eqref{eq:wyd-extension}.

For the observable $A_0$ defined in Eq.~\eqref{eq:equality},
\begin{equation}
V_\rho^{\rm cl}(A_0)=0,\qquad
V_\rho(A_0)=\lM+\lm,
\end{equation}
while
\begin{equation}
I_{\rho,\alpha}(A_0)
=
(\lM^\alpha-\lm^\alpha)
(\lM^{1-\alpha}-\lm^{1-\alpha}).
\end{equation}
Hence Eq.~\eqref{eq:wyd-extension} is saturated by $A_0$, and any increase of the coefficient would violate the inequality for this observable.
Thus the coefficient is optimal.

If $\rho$ is maximally mixed, then $I_{\rho,\alpha}(A)=0$ and $V_\rho^{\rm cl}(A)=V_\rho(A)$ for every observable, so Eq.~\eqref{eq:wyd-extension} is saturated.
\end{proof}

At $\alpha=1/2$, Theorem~\ref{thm:wyd-extension} coincides with the $s=1/2$ case of Theorem~\ref{thm:s-main}, as expected.

\bigskip 

We finally consider the normalized quantum Fisher information $\tilde F_\rho(A)$ defined in Eq.~\eqref{eq:QFI-def}.
The required scalar inequality is the following; its proof is also given in Appendix~\ref{app:lemma}.

\begin{lemma}\label{lem:qfi-scalar}
Let $0\le m<M$.
Then, for any $x,y\in[m,M]$ with $x+y>0$,
\begin{equation}\label{eq:QFI-scalar}
x+y
\ge
\left(\frac{M+m}{M-m}\right)^2
\frac{(x-y)^2}{x+y}.
\end{equation}
\end{lemma}

\begin{theorem}\label{thm:QFI-extension}
For every density operator $\rho$ and every observable $A$,
\begin{equation}\label{eq:QFI-extension}
V_\rho(A)
\ge
V_\rho^{\rm cl}(A)
+
\left(
\frac{\lM+\lm}{\lM-\lm}
\right)^2
\tilde F_\rho(A).
\end{equation}
When $\rho$ is maximally mixed, the second term on the right-hand side is understood to be zero.
For every nonmaximally mixed state $\rho$, the relation is tight.
In particular, the displayed coefficient is optimal: it cannot be replaced by any larger coefficient while preserving Eq.~\eqref{eq:QFI-extension} for every observable $A$.
\end{theorem}

\begin{proof}
Assume first that $\lm<\lM$.
Fix an arbitrary orthonormal eigenbasis $\{|\phi_i\rangle\}_{i=1}^d$ of $\rho$ and set $A_{ij}:=\langle\phi_i|A|\phi_j\rangle$.
By Eq.~\eqref{eq:variance-basis-decomp},
\begin{equation}
V_\rho(A)-V_\rho\left(A^{\rm cl}_{\{\phi_i\}}\right)
=
\sum_{i<j}(\lambda_i+\lambda_j)|A_{ij}|^2.
\end{equation}
Using Eq.~\eqref{eq:QFI-def}, we have
\begin{equation}
\tilde F_\rho(A)
=
\sum_{\substack{i<j\\ \lambda_i+\lambda_j>0}}
\frac{(\lambda_i-\lambda_j)^2}{\lambda_i+\lambda_j}|A_{ij}|^2.
\label{eq:QFI-pair-expansion}
\end{equation}
For every pair $i<j$ with $\lambda_i+\lambda_j>0$, Lemma~\ref{lem:qfi-scalar}, applied with $m=\lm$, $M=\lM$, $x=\lambda_i$, and $y=\lambda_j$, gives
\begin{equation}
\lambda_i+\lambda_j
\ge
\left(
\frac{\lM+\lm}{\lM-\lm}
\right)^2
\frac{(\lambda_i-\lambda_j)^2}{\lambda_i+\lambda_j}.
\end{equation}
Multiplying by $|A_{ij}|^2$ and summing over all such pairs yields
\begin{equation}
V_\rho(A)
\ge
V_\rho\left(A^{\rm cl}_{\{\phi_i\}}\right)
+
\left(
\frac{\lM+\lm}{\lM-\lm}
\right)^2
\tilde F_\rho(A).
\end{equation}
Since the eigenbasis was arbitrary, taking the supremum over all eigenbases of $\rho$ and using Eq.~\eqref{eq:Vcl} proves Eq.~\eqref{eq:QFI-extension}.

For the observable $A_0$ defined in Eq.~\eqref{eq:equality},
\begin{equation}
V_\rho^{\rm cl}(A_0)=0,\qquad
V_\rho(A_0)=\lM+\lm,
\end{equation}
and
\begin{equation}
\tilde F_\rho(A_0)
=
\frac{(\lM-\lm)^2}{\lM+\lm}.
\end{equation}
Hence Eq.~\eqref{eq:QFI-extension} is saturated by $A_0$, and any increase of the coefficient would violate the inequality for this observable.
Thus the coefficient is optimal.

If $\rho$ is maximally mixed, then $\tilde F_\rho(A)=0$ and $V_\rho^{\rm cl}(A)=V_\rho(A)$ for every observable, so Eq.~\eqref{eq:QFI-extension} is saturated.
\end{proof}

Theorems~\ref{thm:s-main}, \ref{thm:wyd-extension}, and \ref{thm:QFI-extension} establish the sharpened relation \eqref{eq:GeneralUR2} for all the families considered in this work.
Since $V_\rho^{\rm cl}(A)\ge0$, dropping the classical contribution immediately yields the corresponding multiplicative bounds of the form \eqref{eq:GeneralUR}.
Moreover, the same observable $A_0$ that establishes tightness of the sharpened relations satisfies $V_\rho^{\rm cl}(A_0)=0$.
Consequently, adding the maximal classical contribution does not reduce the optimal coefficient of the quantum-uncertainty term.

\section{Characterization of universal exactness and qubit systems}
\label{sec:qubit-single}

Having established the sharpened inequalities in the previous section, we now ask a stronger structural question: for which states do these inequalities become exact identities for every observable?  The answer is completely determined by the eigenvalue structure of the state.  Within the class of full-rank nonmaximally mixed states, universal exactness occurs if and only if the state has exactly two distinct eigenvalues.

\begin{theorem}\label{thm:universal-exactness}
Let $\rho$ be a full-rank, nonmaximally mixed density operator, with
$0<\lambda_{\min}<\lambda_{\max}$.
Then the following statements are equivalent.
\begin{enumerate}
\item[(i)]
$\rho$ has exactly two distinct eigenvalues.

\item[(ii)]
For every observable $A$ and every $s\in[1/2,\infty)$,
\begin{equation}
  V_\rho(A)
  =
  V_\rho^{\mathrm{cl}}(A)
  +
  \frac{\lambda_{\max}+\lambda_{\min}}
  {2(\lambda_{\max}^s-\lambda_{\min}^s)^2}
  \norm{\Com{A}{\rho^s}}^2 .
  \label{eq:two-eigenvalue-rhos-identity}
\end{equation}

\item[(iii)]
For every observable $A$ and every $0<\alpha<1$,
\begin{equation}
  V_\rho(A)
  =
  V_\rho^{\mathrm{cl}}(A)
  +
  \frac{\lambda_{\max}+\lambda_{\min}}
  {(\lambda_{\max}^\alpha-\lambda_{\min}^\alpha)
   (\lambda_{\max}^{1-\alpha}-\lambda_{\min}^{1-\alpha})}
  I_{\rho,\alpha}(A).
  \label{eq:two-eigenvalue-wyd-identity}
\end{equation}

\item[(iv)]
For every observable $A$,
\begin{equation}
  V_\rho(A)
  =
  V_\rho^{\mathrm{cl}}(A)
  +
  \left(
  \frac{\lambda_{\max}+\lambda_{\min}}
  {\lambda_{\max}-\lambda_{\min}}
  \right)^2 \tilde{F}{\rho}(A).
  \label{eq:two-eigenvalue-qfi-identity}
\end{equation}
\end{enumerate}
In particular, the two-eigenvalue condition is necessary and sufficient separately for universal exactness of each of the three sharpened families.
\end{theorem}

\begin{proof}
We first prove the sufficiency, $(\mathrm{i})\Rightarrow(\mathrm{ii})$, $(\mathrm{iii})$, and $(\mathrm{iv})$.  Assume that $\rho$ has exactly two distinct eigenvalues.  Write its spectral decomposition as
\begin{equation}
  \rho
  =
  \lambda_- P_-
  +
  \lambda_+ P_+,
  \qquad
  0<\lambda_-<\lambda_+,
  \qquad
  P_-+P_+=\I,
  \label{eq:two-ev-st}
\end{equation}
where $\lambda_-=\lambda_{\min}$ and $\lambda_+=\lambda_{\max}$.  The ranks of $P_-$ and $P_+$ are arbitrary.  Hence there is only one pair of eigenspaces carrying different eigenvalues, although each eigenspace may have arbitrary dimension.

Set
\[
  X:=P_-AP_+ .
\]
Since $A$ is self-adjoint, the opposite off-diagonal block is
\[
  P_+AP_-=X^\dagger .
\]
Proposition~\ref{prop:classical-pinching} gives
\[
  V_\rho^{\mathrm{cl}}(A)
  =
  V_\rho(D),
\]
where 
\[
D:=\mathcal P_\rho(A)
   =P_-AP_-+P_+AP_+.
\]
Note that 
\[
A=D+X+X^\dagger.
\]
Expanding $A^2$ around the block-diagonal part $D$ gives
\begin{align*}
A^2-D^2
&=
D(X+X^\dagger)
+(X+X^\dagger)D
+(X+X^\dagger)^2.
\end{align*}
Both $\rho$ and $D$ are block diagonal with respect to
$\operatorname{Ran}P_-\oplus\operatorname{Ran}P_+$, whereas $X$ and $X^\dagger$ are off diagonal.  Therefore the trace of every product containing exactly one off-diagonal block vanishes, and hence
\[
\Tr(\rho A)=\Tr(\rho D),
\qquad
\Tr\left(\rho\{D,X+X^\dagger\}\right)=0.
\]
Moreover, $P_-P_+=0$ implies
\[
X^2=(X^\dagger)^2=0,
\]
so that
\[
(X+X^\dagger)^2
=
XX^\dagger+X^\dagger X.
\]
Using these facts in the definition of the variance, and using
$V_\rho^{\mathrm{cl}}(A)=V_\rho(D)$, we obtain
\begin{align}
V_\rho(A)-V_\rho^{\mathrm{cl}}(A)
&=
\Tr\left[\rho(XX^\dagger+X^\dagger X)\right]
\nonumber\\
&=
\lambda_-\Tr(XX^\dagger)
+\lambda_+\Tr(X^\dagger X)
\nonumber\\
&=
(\lambda_-+\lambda_+)\|X\|^2.
\label{eq:two-ev-off-var}
\end{align}

We next evaluate the three noncommutative quantities appearing in Theorems~\ref{thm:s-main}, \ref{thm:wyd-extension}, and~\ref{thm:QFI-extension}.  Since
\[
\rho^s=\lambda_-^sP_-+\lambda_+^sP_+,
\]
we have
\[
X\rho^s=\lambda_+^sX,
\qquad
\rho^sX=\lambda_-^sX,
\]
and therefore
\[
[X,\rho^s]
=
(\lambda_+^s-\lambda_-^s)X.
\]
Similarly,
\[
[X^\dagger,\rho^s]
=
-(\lambda_+^s-\lambda_-^s)X^\dagger.
\]
Because $[D,\rho^s]=0$, it follows that
\[
[A,\rho^s]
=
(\lambda_+^s-\lambda_-^s)(X-X^\dagger).
\]
Using $X^2=(X^\dagger)^2=0$ once again,
\begin{align*}
\|X-X^\dagger\|^2
&=
\Tr\left[(X-X^\dagger)^\dagger(X-X^\dagger)\right]\\
&=
\Tr(X^\dagger X+XX^\dagger)\\
&=
2\|X\|^2.
\end{align*}
Consequently,
\begin{equation}
\|[A,\rho^s]\|^2
=
2(\lambda_+^s-\lambda_-^s)^2\|X\|^2.
\label{eq:two-eigenvalue-commutator}
\end{equation}
Substituting Eq.~\eqref{eq:two-eigenvalue-commutator} into Eq.~\eqref{eq:two-ev-off-var} shows that equality holds in Eq.~\eqref{eq:two-eigenvalue-rhos-identity}.

This proves $(\mathrm{i})\Rightarrow(\mathrm{ii})$.
For the Wigner--Yanase--Dyson skew information,
\[
[\rho^\alpha,A]
=
(\lambda_+^\alpha-\lambda_-^\alpha)(X^\dagger-X),
\]
and
\[
[\rho^{1-\alpha},A]
=
(\lambda_+^{1-\alpha}-\lambda_-^{1-\alpha})(X^\dagger-X).
\]
Therefore,
\begin{align}
[\rho^\alpha,A][\rho^{1-\alpha},A]
&=
(\lambda_+^\alpha-\lambda_-^\alpha)
(\lambda_+^{1-\alpha}-\lambda_-^{1-\alpha})
(X^\dagger-X)^2
\nonumber\\
&=
-(\lambda_+^\alpha-\lambda_-^\alpha)
(\lambda_+^{1-\alpha}-\lambda_-^{1-\alpha})
(X^\dagger X+XX^\dagger),
\end{align}
where the second equality again uses $X^2=(X^\dagger)^2=0$.  By the definition of $I_{\rho,\alpha}(A)$,
\begin{align}
I_{\rho,\alpha}(A)
&=
-\frac12\Tr\left([\rho^\alpha,A][\rho^{1-\alpha},A]\right)
\nonumber\\
&=
(\lambda_+^\alpha-\lambda_-^\alpha)
(\lambda_+^{1-\alpha}-\lambda_-^{1-\alpha})
\|X\|^2.
\label{eq:two-eigenvalue-wyd-block}
\end{align}
Combining Eq.~\eqref{eq:two-eigenvalue-wyd-block} with Eq.~\eqref{eq:two-ev-off-var} shows the equality in \eqref{eq:two-eigenvalue-wyd-identity}. 
Thus $(\mathrm{i})\Rightarrow(\mathrm{iii})$.

Finally, for the quantum Fisher information, the only terms with distinct eigenvalues are the two off-diagonal blocks $P_-AP_+$ and $P_+AP_-$. 
From Eq.~\eqref{eq:QFI-def}, the two ordered contributions combine to give
\begin{equation}
\tilde{F}_{\rho}(A)
=
\frac{(\lambda_+-\lambda_-)^2}{\lambda_++\lambda_-}\|X\|^2.
\label{eq:two-eigenvalue-qfi-block}
\end{equation}
Substituting Eq.~\eqref{eq:two-eigenvalue-qfi-block} into Eq.~\eqref{eq:two-ev-off-var} gives equality \eqref{eq:two-eigenvalue-qfi-identity}. 
Hence $(\mathrm{i})\Rightarrow(\mathrm{iv})$.

We now prove the converse implications.  Because $\rho$ is assumed nonmaximally mixed, it has at least two distinct eigenvalues.
If condition~$(\mathrm{i})$ fails, it therefore has at least three distinct eigenvalues.
Choose an eigenvalue $\mu$ satisfying
\[
0<m:=\lambda_{\min}<\mu<M:=\lambda_{\max}.
\]
Let $|m\rangle$ and $|\mu\rangle$ be unit eigenvectors belonging to $m$ and $\mu$, respectively, and define the Hermitian observable
\[
A_{m\mu}:=|m\rangle\langle\mu|+|\mu\rangle\langle m|.
\]
This observable has no block-diagonal part with respect to the spectral decomposition of $\rho$.  Hence $\mathcal P_\rho(A_{m\mu})=0$ and
\[
V_\rho^{\rm cl}(A_{m\mu})=0.
\]
Moreover, $\Tr(\rho A_{m\mu})=0$ and
\[
A_{m\mu}^2
=|m\rangle\langle m|+|\mu\rangle\langle\mu|,
\]
so
\[
V_\rho(A_{m\mu})
=\Tr(\rho A_{m\mu}^2)
=m+\mu.
\]

For the power-commutator family, direct evaluation gives
\[
\|[A_{m\mu},\rho^s]\|^2
=2(\mu^s-m^s)^2.
\]
Therefore the noncommutative term in the sharpened inequality of Theorem~\ref{thm:s-main} equals
\[
\frac{M+m}{(M^s-m^s)^2}(\mu^s-m^s)^2.
\]
To compare this quantity with $m+\mu$, we use the strictness contained in the proof of Lemma~\ref{lem:scalar-ieq} given in Appendix \ref{app:lemma}.  
In that proof, we introduce the variables
\[
t=\frac{x}{y},
\qquad
R=\frac{M}{m}.
\]
For the present pair $(x,y)=(\mu,m)$, where $m<\mu<M$,
\[
1<t=\frac{\mu}{m}<\frac{M}{m}=R.
\]
The functions $t/(1+t)$, $1-t^{-p}$, and $1-t^{-q}$ appearing in Eq.~\eqref{eq:gen-rewrite} are all strictly increasing for $t>1$.
Hence the inequality in Lemma~\ref{lem:scalar-ieq} is strict whenever the pair of eigenvalues is $(\mu,m)$ rather than the minimum--maximum pair $(M,m)$.

Taking $(p,q,\gamma)=(s,s,1)$ in Lemma~\ref{lem:scalar-ieq} therefore yields
\[
m+\mu
>
\frac{M+m}{(M^s-m^s)^2}(\mu^s-m^s)^2.
\]
Thus, for every fixed $s\in[1/2,\infty)$, the sharpened power-commutator relation is strict for the observable $A_{m\mu}$.  In particular, condition~$(\mathrm{ii})$ cannot hold.

Taking $(p,q,\gamma)=(\alpha,1-\alpha,1)$ gives
\[
m+\mu
>
\frac{M+m}
{(M^\alpha-m^\alpha)(M^{1-\alpha}-m^{1-\alpha})}
(\mu^\alpha-m^\alpha)(\mu^{1-\alpha}-m^{1-\alpha}).
\]
The right-hand side is exactly the sharpened Wigner--Yanase--Dyson term for $A_{m\mu}$.  Hence the WYD relation is strict for every $0<\alpha<1$, and condition~$(\mathrm{iii})$ cannot hold.

Finally, taking $(p,q,\gamma)=(1,1,2)$ gives the strict inequality
\[
(m+\mu)^2
>
\left(\frac{M+m}{M-m}\right)^2(\mu-m)^2.
\]
Since $m+\mu>0$, division by $m+\mu$ yields
\[
m+\mu
>
\left(\frac{M+m}{M-m}\right)^2
\frac{(\mu-m)^2}{\mu+m}.
\]
For $A_{m\mu}$, Eq.~\eqref{eq:QFI-def} gives
\[
\tilde{F}_{\rho}(A)
=
\frac{(\mu-m)^2}{\mu+m}.
\]
Thus the sharpened QFI relation is also strict, so condition~$(\mathrm{iv})$ cannot hold.

We have shown that failure of condition~$(\mathrm{i})$ forces failure of each of conditions~$(\mathrm{ii})$--$(\mathrm{iv})$.  Equivalently, each of $(\mathrm{ii})$, $(\mathrm{iii})$, and $(\mathrm{iv})$ implies $(\mathrm{i})$.  Together with the sufficiency proved above, this establishes the equivalence of all four statements.
\end{proof}

Theorem~\ref{thm:universal-exactness} gives a complete characterization rather than only a sufficient condition.
Universal exactness is possible in arbitrary finite dimension, but only when all noncommuting contributions of an observable connect the same two eigenspaces associated with different eigenvalues.
Degeneracy within either eigenspace does not affect the result.
Conversely, if a third distinct eigenvalue is present, one can choose a pair of distinct eigenvalues different from $\{\lambda_{\min},\lambda_{\max}\}$.
The strictness of the unified scalar inequality in Appendix~\ref{app:lemma} then prevents universal saturation.

Having established that universal exactness is characterized by the two-eigenvalue condition, we now specialize to qubit systems.
This case allows us to analyze the unsharpened relations and their equality conditions explicitly.

Every full-rank nonmaximally mixed qubit state has exactly two distinct eigenvalues.  Therefore, Theorem~\ref{thm:universal-exactness} immediately implies the exact decomposition for this class.
The maximally mixed and pure cases will be treated separately below.
Nevertheless, we give the explicit Bloch-sphere expressions because they also clarify the unsharpened bounds and their equality conditions.

To make these statements explicit, let $\bm{\sigma}:=(\sigma_x,\sigma_y,\sigma_z)$ be the vector of Pauli matrices.
Let the qubit state be represented by the Bloch vector $r\bm{n}$, with length $0\le r\le 1$ and direction $\rnorm{\bm{n}}=1$, as
\begin{equation}
  \rho=\frac{\id+r\bm{n}\cdot\bm{\sigma}}{2}.
  \nonumber
\end{equation}
Here, $r=1$ gives a pure state, while $r=0$ gives the maximally mixed state.

For this state, we consider the spin observable in the direction $\bm a$:
\begin{equation}
  A=\bm{a}\cdot\bm{\sigma},
  \qquad
  \rnorm{\bm{a}}=1.\nonumber
\end{equation}
Since $A^2=\I$ and $\langle A\rangle_\rho=r\bm{a}\cdot\bm{n}$, it follows that
\begin{equation}
  V_\rho(A)
  =
  1-r^2(\bm{a}\cdot\bm{n})^2 .
  \label{eq:qubit-variance}
\end{equation}

To express this result geometrically, let $\theta\in[0,\pi]$ be the angle between $\bm a$ and $\bm n$.
Then $\bm a\cdot\bm n=\cos\theta$, and hence
\begin{equation}
  V_\rho(A)=1-r^2\cos^2\theta .
  \label{eq:qubit-var-theta}
\end{equation}
For $0<r\le 1$, the eigenvalues of $\rho$ are $\lM=\lambda_2 = (1+r)/2$ and $\lm=\lambda_1 = (1-r)/2$.
In the eigenbasis of $\rho$, one has $\abs{A_{12}}^2=\sin^2\theta$.
Therefore, for any $s\in[1/2,\infty)$,
\[
  \norm{\Com{A}{\rho^s}}^2
  =
  2(\lM^s-\lm^s)^2\sin^2\theta .
\]
Since $\lM+\lm=1$, the optimal coefficient \eqref{eq:optCoef}
\[
c^{\mathrm{opt}}_s(\rho)
=
\frac{\lM+\lm}{2(\lM^s-\lm^s)^2}
=
\frac{1}{2(\lM^s-\lm^s)^2}
\]
satisfies
\begin{equation}
c^{\mathrm{opt}}_s(\rho)
\norm{\Com{A}{\rho^s}}^2
=
\sin^2\theta. 
\label{eq:qubit-s-bound}
\end{equation}
The same collapse occurs for the Wigner--Yanase--Dyson and quantum-Fisher-information families.  Indeed, for $0<r\le 1$,
\begin{align}
I_{\rho,\alpha}(A)
&=
(\lM^\alpha-\lm^\alpha)
(\lM^{1-\alpha}-\lm^{1-\alpha})
\sin^2\theta,
\label{eq:qubit-wyd-block-explicit}\\
\frac{1}{4}F_Q[\rho,A]
&=
\frac{(\lM-\lm)^2}{\lM+\lm}\sin^2\theta
=
r^2\sin^2\theta.
\label{eq:qubit-qfi-block-explicit}
\end{align}
Since $\lM+\lm=1$ and $\lM-\lm=r$, the optimized WYD contribution and the optimized QFI contribution satisfy
\begin{align}
\frac{I_{\rho,\alpha}(A)}
{(\lM^\alpha-\lm^\alpha)(\lM^{1-\alpha}-\lm^{1-\alpha})}
&=\sin^2\theta,
\label{eq:qubit-wyd-opt-explicit}\\
\left(\frac{\lM+\lm}{\lM-\lm}\right)^2
\tilde{F}_\rho(A)
&=\sin^2\theta.
\label{eq:qubit-qfi-opt-explicit}
\end{align}
Thus, for every nonmaximally mixed qubit state, the optimized noncommutative terms in Theorems~\ref{thm:s-main}, \ref{thm:wyd-extension}, and~\ref{thm:QFI-extension} all reduce to the same geometric quantity $\sin^2\theta$.
The optimized noncommutative contribution in the power-commutator family is independent of $s$.
Thus, in the qubit case, all the optimized lower bounds in this family coincide.
In particular, the bounds \eqref{eq:WY-opt-intro} and \eqref{eq:rho-opt-intro} yield the same relation:
\begin{equation}
  V_\rho(A)
  =
  1-r^2\cos^2\theta
  \ge
  \sin^2\theta.
\end{equation}
The equality condition is thus given by 
\begin{equation}
  r=1
  \qquad
  \text{or}
  \qquad
  \theta=\frac{\pi}{2}.
\end{equation}
This result admits the following interpretation.
For pure states, the lower bound obtained in this work coincides with the variance itself for every observable $A$.
In other words, the variance in a pure qubit state is completely characterized by the noncommutativity between the state and the observable.
For mixed states with $0<r<1$, by contrast, equality holds only when $\bm{a}\perp\bm{n}$, that is, only when $A$ has a purely off-diagonal component in the eigenbasis of $\rho$.
In that case, $\langle A\rangle_\rho=0$ and $V_\rho(A)=1$.
Thus, in mixed states, equality does not correspond to the case where the observable becomes definite; rather, it corresponds to the case where the state is maximally uncertain with respect to that observable.

This contrast also clarifies the role of the three sharpened relations,  in which the classical contribution to the variance is kept explicitly. 
In the present qubit setting, a direct computation shows that the classical uncertainty \eqref{eq:Vcl} is given by
\begin{equation}
  V_\rho^{\mathrm{cl}}(A)
  =
  (1-r^2)\cos^2\theta .
\end{equation}

Consequently, all three sharpened relations become exact identities in the qubit case.
For $0<r<1$, this is also a direct consequence of Theorem~\ref{thm:universal-exactness}.
The pure-state case $r=1$ follows from the explicit calculation above.
For $r=0$, the state is maximally mixed. In this case $\mathcal P_\rho(A)=A$, while all three noncommutative terms vanish.
Thus all three exact identities extend to every qubit state.

Every qubit observable can be written as $A=a_0I+\bm{a}\cdot\bm{\sigma}$. Both sides of the identity are invariant under adding $a_0I$ and are homogeneous under rescaling of $\bm{a}$. Hence the result extends to an arbitrary qubit observable.

\begin{proposition}\label{prop:qubit-identity}
Let $\rho$ be any qubit state and let $A$ be any qubit observable.
For every $s\in[1/2,\infty)$,
\begin{equation}\label{eq:mainURQubit}
 V_\rho(A)
 =
 V_\rho^{\mathrm{cl}}(A)
 +
 \frac{1}{2(\lM^s-\lm^s)^2}
 \norm{\Com{A}{\rho^s}}^2.
\end{equation}
For the maximally mixed state, the second term on the right-hand side is understood to be zero.
\end{proposition}

The same conclusion holds for the Wigner--Yanase--Dyson and quantum-Fisher-information families. For full-rank nonmaximally mixed qubit states, this follows directly from Theorem \ref{thm:universal-exactness}; the pure and maximally mixed cases follow from the explicit calculations above.

\begin{proposition}\label{prop:qubit-wyd-identity}
Let $\rho$ be any qubit state and let $A$ be any qubit observable.
For every $0<\alpha<1$,
\begin{equation}\label{eq:qubit-wyd-identity}
  V_\rho(A)
  =
  V_\rho^{\mathrm{cl}}(A)
  +
  \frac{1}{(\lM^\alpha-\lm^\alpha)(\lM^{1-\alpha}-\lm^{1-\alpha})}
  I_{\rho,\alpha}(A).
\end{equation}
For the maximally mixed state, the second term on the right-hand side is understood to be zero.
\end{proposition}

\begin{proposition}\label{prop:qubit-qfi-identity}
Let $\rho$ be any qubit state and let $A$ be any qubit observable.
If $\rho$ is not maximally mixed, then
\begin{equation}\label{eq:qubit-qfi-identity}
  V_\rho(A)
  =
  V_\rho^{\mathrm{cl}}(A)
  +
  \frac{1}{(\lM-\lm)^2}\tilde{F}_{\rho}(A)
\end{equation}
For the maximally mixed state, the second term on the right-hand side is understood to be zero.
\end{proposition}

The exactness found above follows directly from the two-eigenvalue structure of the state.
For a full-rank state with exactly two distinct eigenvalues, there is only one pair of eigenspaces associated with different eigenvalues. 
Consequently, every noncommuting contribution of an observable is carried by the same off-diagonal block connecting these two eigenspaces. 
Both the variance beyond its classical part and the commutator-based quantities are therefore determined by this single block.
This is why the sharpened inequalities become exact, regardless of the dimensions of the two eigenspaces. 
Qubit systems provide the simplest realization of this structure.

In qubit systems, Propositions~\ref{prop:qubit-identity}, \ref{prop:qubit-wyd-identity}, and~\ref{prop:qubit-qfi-identity} show that the sharpened power-commutator, Wigner--Yanase--Dyson, and quantum-Fisher-information relations are exact for every state.
They give a complete decomposition of the variance into a classical diagonal contribution and a noncommutative off-diagonal contribution.
The Bloch-sphere representation further shows how the unsharpened bounds behave and when they are saturated.

\section{Application to variance-product trade-offs and comparison in qubit systems}\label{sec:product}

Having established \emph{single-observable preparation uncertainty relations} as a self-contained framework, we now ask what they imply for conventional multi-observable trade-offs.

As a reference point, applying Luo's inequality to two observables $A$ and $B$ gives

\begin{equation}
V_\rho(A) V_\rho(B)
\ge
\frac{1}{4}
\|[A,\sqrt{\rho}]\|^2\|[B,\sqrt{\rho}]\|^2. 
\label{Luo2}
\end{equation}

On the other hand, dropping the nonnegative classical term from Eq.~\eqref{eq:mainUR}, together with the optimal coefficient \eqref{eq:optCoef}, yields the following product lower bound for $V_\rho(A)V_\rho(B)$:

\begin{align}
&
\frac{1}{4}
\frac{(\lM+\lm)^2}{(\lM^s-\lm^s)^4}
\|[A,\rho^s]\|^2\|[B,\rho^s]\|^2,
\label{Our1} 
\end{align}
and the sharper relation \eqref{eq:mainUR} gives the bound 
\begin{align}
&
\left(
  V_\rho^{\mathrm{cl}}(A)
  +
  \frac{\lM+\lm}
  {2(\lM^s-\lm^s)^2}
  \|[A,\rho^s]\|^2
\right) \nonumber \\
&\qquad\times
\left(
  V_\rho^{\mathrm{cl}}(B)
  +
  \frac{\lM+\lm}
  {2(\lM^s-\lm^s)^2}
  \|[B,\rho^s]\|^2
\right).
\label{Our1Sharp}
\end{align}
For the maximally mixed state, the power-commutator terms in Eqs.~\eqref{Our1} and \eqref{Our1Sharp} are understood to vanish, as in Theorem~\ref{thm:s-main}.

Similarly, Theorem~\ref{thm:wyd-extension} gives the
Wigner--Yanase--Dyson product-form bound
\begin{align}
V_\rho(A)V_\rho(B)
\ge&
\left[
V_\rho^{\mathrm{cl}}(A)
+
\frac{\lM+\lm}
{(\lM^\alpha-\lm^\alpha)(\lM^{1-\alpha}-\lm^{1-\alpha})}
I_{\rho,\alpha}(A)
\right]
\nonumber\\
&\times
\left[
V_\rho^{\mathrm{cl}}(B)
+
\frac{\lM+\lm}
{(\lM^\alpha-\lm^\alpha)(\lM^{1-\alpha}-\lm^{1-\alpha})}
I_{\rho,\alpha}(B)
\right].
\label{OurWYDSharp}
\end{align}
For the maximally mixed state, the Wigner--Yanase--Dyson terms
are understood to vanish, as in Theorem~\ref{thm:wyd-extension}.
Dropping the classical contributions gives
\begin{equation}
V_\rho(A)V_\rho(B)
\ge
\left[
\frac{\lM+\lm}
{(\lM^\alpha-\lm^\alpha)(\lM^{1-\alpha}-\lm^{1-\alpha})}
\right]^2
I_{\rho,\alpha}(A)I_{\rho,\alpha}(B).
\label{OurWYD}
\end{equation}

    Similarly, Theorem~\ref{thm:QFI-extension} gives the
quantum-Fisher-information product-form bound
\begin{align}
V_\rho(A)V_\rho(B)
\ge&
\left[
  V_\rho^{\mathrm{cl}}(A)
  +
  \left(
    \frac{\lM+\lm}{\lM-\lm}
  \right)^2
  \tilde{F}_\rho(A)
\right]
\nonumber\\
&\times
\left[
  V_\rho^{\mathrm{cl}}(B)
  +
  \left(
    \frac{\lM+\lm}{\lM-\lm}
  \right)^2
  \tilde{F}_\rho(B)
\right].
\label{OurQFISharp}
\end{align}
For the maximally mixed state, the quantum-Fisher-information terms
are understood to vanish, as in Theorem~\ref{thm:QFI-extension}.
Dropping the classical contributions gives
\begin{equation}
V_\rho(A)V_\rho(B)
\ge
\left(
\frac{\lM+\lm}{\lM-\lm}
\right)^4
\tilde{F}_\rho(A)\tilde{F}_\rho(B).
\label{OurQFI}
\end{equation}

These lower bounds are obtained by combining single-observable uncertainty estimates.
They do not explicitly use the noncommutativity between the observables themselves.
One might therefore expect them to be weaker than conventional trade-off relations.
The qubit case, however, shows that this expectation need not be correct. Remarkably, the sharpened bounds become exact.
Even the bound that retains only the noncommutative contributions is stronger than the Robertson bound for every pair of qubit spin observables.
The comparison with the Schr\"odinger bound is different.
Neither bound is always stronger for every pair.
For mixed qubit states, however, our bound is stronger after uniform averaging over all pairs of spin observables.

For comparison, and to make the distinction explicit, let us recall these conventional relations.

For two observables $A$ and $B$, Robertson's relation states that
\begin{equation}
  V_\rho(A)V_\rho(B)
  \ge
  \frac{1}{4}\abs{\langle [A,B]\rangle_\rho}^2,
  \label{Rob}
\end{equation}
where $\langle X\rangle_\rho:=\Tr(\rho X)$.
The refined Schr\"odinger relation is given by
\begin{equation}
  V_\rho(A)V_\rho(B)
  \ge
  \frac{1}{4}\abs{\langle [A,B]\rangle_\rho}^2
  +
  \mathrm{Cov}_\rho(A,B)^2,
  \label{Sch}
\end{equation}
where
\begin{equation}
  \mathrm{Cov}_\rho(A,B)
  :=
  \frac{1}{2}\langle \{A,B\}\rangle_\rho
  -
  \langle A\rangle_\rho\langle B\rangle_\rho.
\end{equation}
Here $\{A,B\}:=AB+BA$ denotes the anticommutator.

Since the comparison below is restricted to qubit systems, the optimized Wigner--Yanase--Dyson and quantum-Fisher-information bounds need not be treated separately.
As shown in the previous section, for every nonmaximally mixed qubit state, the optimized noncommutative contributions of the power-commutator, Wigner--Yanase--Dyson, and quantum-Fisher-information families all reduce to the same quantity.
Consequently, their corresponding product-form lower bounds also coincide for qubits.
We therefore use the power-commutator bound as a representative
of these three equivalent families in the comparison below.

For convenience, we introduce notation for the five lower bounds.
The right-hand sides of Eqs.~\eqref{Rob}, \eqref{Sch}, \eqref{Luo2}, \eqref{Our1}, and \eqref{Our1Sharp} are denoted by
\[
B_{\rm R},\qquad
B_{\rm S},\qquad
B_{\rm L},\qquad
B_\rho^{\rm opt},\qquad
B_\rho^{\rm sharp},
\]
respectively.

Let us now compare these bounds in the qubit case.
The previous section showed that the sharpened single-observable relations are exact identities for qubits.
This includes the classical contribution.

Consequently, the sharpened power-commutator, Wigner--Yanase--Dyson, and quantum-Fisher-information product bounds are all exact identities for the product of variances.
Indeed, this follows directly from Propositions~\ref{prop:qubit-identity}, \ref{prop:qubit-wyd-identity}, and~\ref{prop:qubit-qfi-identity}.
Once the classical contributions are retained, these product relations therefore leave no room for further improvement.

This exact attainability is in the same spirit as the corresponding results obtained in Hayashi's relation \cite{Hayashi2017} and in our previous work \cite{KMOLC}.

\emph{More surprisingly}, $B_\rho^{\rm opt}$ retains only the noncommutative contributions but is still no weaker than the Robertson bound.
This holds for every nonmaximally mixed qubit state and every pair of spin observables.
To see this explicitly, use rotational invariance to choose the Bloch vector along the $z$-axis:
\[
\rho=\frac{I+r\sigma_z}{2},
\qquad 0<r\le 1.
\]
Let
\[
{\bm a}=(a_1,a_2,a_3),\qquad
{\bm b}=(b_1,b_2,b_3),
\]
and
\[
A={\bm a}\cdot{\bm \sigma},
\qquad
B={\bm b}\cdot{\bm \sigma},
\qquad
|{\bm a}|=|{\bm b}|=1.
\]
Using the qubit expressions derived in Appendix~\ref{app:qubit-average}, the Robertson bound is
\begin{equation}
B_{\rm R}
=
r^2(a_1b_2-a_2b_1)^2,
\label{eq:BR-qubit}
\end{equation}
whereas the optimized single-observable bound \eqref{Our1} reduces to
\begin{equation}
B_\rho^{\rm opt}
=
(1-a_3^2)(1-b_3^2).
\label{eq:Bopt-qubit}
\end{equation}
Since $0<r\le 1$, the Cauchy--Schwarz inequality gives 
\begin{align} 
B_{\rm R} &= r^2(a_1b_2-a_2b_1)^2 \nonumber\\ 
&\le r^2(a_1^2+a_2^2)(b_1^2+b_2^2) \nonumber\\ 
&= r^2 B_\rho^{\rm opt} \le B_\rho^{\rm opt}. \label{eq:Robertson-pointwise-comparison} \end{align} 
Hence, 
\begin{equation} 
B_{\rm R} \le r^2 B_\rho^{\rm opt} \le B_\rho^{\rm opt} \label{eq:Robertson-pointwise} 
\end{equation} 
for every pair of qubit spin observables.
Thus $B_\rho^{\rm opt}$ already bounds the Robertson contribution from above.
This is notable because $B_\rho^{\rm opt}$ is obtained by applying a single-observable relation independently to $A$ and $B$.
The Robertson contribution instead depends explicitly on the two-observable commutator.
For $0<r<1$, the comparison is strict whenever $B_\rho^{\rm opt}>0$.

The situation is different for the Schr\"odinger bound.
Neither $B_\rho^{\rm opt}$ nor $B_{\rm S}$ is always stronger than the other for every pair of observables.
We therefore compare their overall strengths by averaging the bounds uniformly over all pairs of qubit spin observables.
Direct computation gives

\begin{align}
\langle B_{\rm R}\rangle_{\rm av}
&=
\frac{2}{9}(2P-1),
\label{eq:avg-R}\\
\langle B_{\rm S}\rangle_{\rm av}
&=
\frac{4}{9}(P^2-P+1),
\label{eq:avg-S}\\
\langle B_{\rm L}\rangle_{\rm av}
&=
\frac{4}{9}\left(1-\sqrt{2(1-P)}\right)^2,
\label{eq:avg-Luo}\\
\langle B_{\rho}^{\rm opt}\rangle _{\rm av}
&=
\frac{4}{9},\label{eq:avg-ours} \\
\langle B_{\rho}^{\rm sharp}\rangle_{\rm av}
&=
\frac{4(2-P)^2}{9}.
\label{eq:avg-ours-sharp}
\end{align}
 Here \(\langle \cdot \rangle_{\rm av}\) denotes the uniform average over two independent unit vectors ${\bm a},{\bm b}\in S^2$.
Equation~\eqref{eq:avg-ours} is stated for nonmaximally mixed qubit states with \(1/2<P\le1\).
The derivation of Eqs.~\eqref{eq:avg-R}--\eqref{eq:avg-ours-sharp} is given in Appendix~\ref{app:qubit-average}.

In particular, Eq.~\eqref{eq:avg-ours} shows that the averaged optimized bound is independent of the purity throughout \(1/2< P\le1\). 
By contrast, the Robertson contribution is suppressed by \(r^2=2P-1\).
Equation~\eqref{eq:Robertson-pointwise} shows this for every pair of qubit spin observables.
For mixed qubit states with \(1/2<P<1\), the averaged optimized bound also exceeds the corresponding averaged Schr\"odinger and Luo-type bounds.
This comparison is shown in Fig.~\ref{fig:comp}.
The advantage over the Robertson bound is therefore not merely an averaged effect.
The stronger relation \(B_{\rm R}\le r^2B_\rho^{\rm opt}\) holds for every pair, and the difference grows as the state becomes more mixed.

For completeness, we should also compare with relations that are already exact in the qubit setting. In particular, the Nagaoka--Hayashi relation~\cite{Nagaoka2005CRbound,Hayashi2017} and the relation~\cite{KMOLC} do not merely give
lower bounds for the product of variances for arbitrary states and
observables; in fact, they hold as equalities.  Hence, from the viewpoint
of product-of-variances relations alone, there is no room to improve these
relations further.

Nevertheless, the present approach reveals a distinction that is not apparent from conventional two-observable uncertainty relations alone.
The Robertson relation describes the uncertainty product through the state-dependent term
\(\frac{1}{4}\abs{\langle[A,B]\rangle_\rho}^2\).
This term depends jointly on $A$ and $B$.
For qubits, Eq.~\eqref{eq:Robertson-pointwise} shows that it is already controlled by the noncommutative contributions associated independently with $A$ and $B$.
The Robertson contribution is bounded by $r^2B_\rho^{\rm opt}$.
By contrast, $B_\rho^{\rm opt}$ itself is independent of $r$.
As the state becomes more mixed, the Robertson commutator term can therefore become small while the optimized single-observable contributions remain unchanged.
The single-observable formulation thus retains information about \emph{state--observable noncommutativity} that need not be visible in the conventional two-observable commutator term.

\begin{figure}[tb]
\centering
\IfFileExists{comp.png}{%
\includegraphics[width=1.0\linewidth]{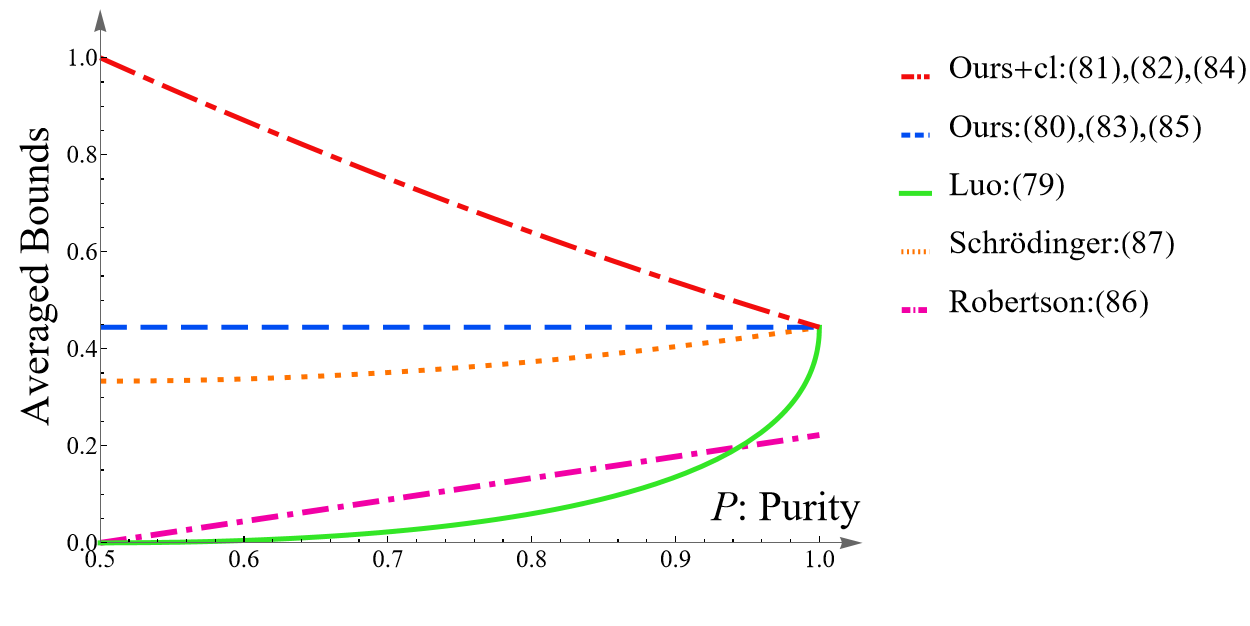}%
}{%
\fbox{\parbox{1.0\linewidth}{\centering comp.png}}%
}
\caption{
State purity versus the uniformly averaged lower bounds obtained from the
Robertson relation \eqref{Rob}, the Schr\"odinger relation \eqref{Sch},
the Luo-type relation \eqref{Luo2}, and the optimal relations
\eqref{Our1}, \eqref{OurWYD}, and \eqref{OurQFI}, together with their
sharpened counterparts \eqref{Our1Sharp}, \eqref{OurWYDSharp}, and
\eqref{OurQFISharp}.
For qubits, the three optimal relations coincide, and likewise the three
sharpened relations coincide.
The corresponding averaged expressions plotted here are given in
Eqs.~\eqref{eq:avg-R}--\eqref{eq:avg-ours-sharp}.
All curves are shown for $1/2<P\leq 1$.
}
\label{fig:comp}
\end{figure}

\section{Conclusion}\label{sec:conclusion}

We have introduced and developed optimal \emph{single-observable preparation uncertainty relations} as a framework complementary to conventional multi-observable trade-off relations.
In this framework, the uncertainty of an observable can be constrained without introducing a second observable; instead, the relevant incompatibility is that between the observable and the quantum state.
More specifically, for a fixed state $\rho$, we derived a one-parameter family of sharp lower bounds on $V_\rho(A)$ in terms of the power-commutators $[A,\rho^s]$, valid for all $s\ge1/2$.
Remarkably, the optimal coefficient is determined solely by the smallest and largest eigenvalues of $\rho$.
For $1/2\le s\le1$, the corresponding quantity $K_{\rho,s}(A)$ qualifies as a quantum uncertainty in the sense considered here, whereas for $s>1$ it is not convex in $\rho$ in general.
Nevertheless, the uncertainty relation itself continues to hold throughout the entire range $s\ge1/2$.
In particular, the case $s=1/2$ gives an optimal state-dependent improvement of Luo's relation based on the Wigner--Yanase skew information, while the case $s=1$ yields the corresponding optimal relation in terms of the direct commutator $[A,\rho]$.

The direct-commutator relation makes the role of state--observable incompatibility particularly transparent: nonzero $A$-coherence or $A$-asymmetry, as detected by $[A,\rho]$, necessarily implies nonzero uncertainty in $A$.
We also established optimal state-dependent bounds for the full Wigner--Yanase--Dyson skew-information family and for the normalized quantum Fisher information $\tilde F_\rho(A)$.
At $\alpha=1/2$, the Wigner--Yanase--Dyson bound reduces to the optimized Wigner--Yanase relation.
For the quantum Fisher information, the optimal coefficient multiplying $\tilde F_\rho(A)$ is
\begin{equation}
\left(\frac{\lM+\lm}{\lM-\lm}\right)^2.
\end{equation}
This result is particularly noteworthy because $\tilde F_\rho(A)$ is already the largest convex extension of the pure-state variance.
The stronger fixed-state bound becomes possible because the optimization is performed for a fixed $\rho$, and the resulting state-dependent rescaling is not required to remain within the same convex class.

Beyond these optimal multiplicative bounds, we further sharpened the relations by introducing the classical uncertainty $V_\rho^{\rm cl}(A)$ as the maximal contribution to the variance arising from the part of the observable compatible with the state.
We showed that this quantity admits a basis-independent representation through the pinching associated with the spectral decomposition of $\rho$.
The resulting uncertainty relations take the form
\begin{equation}
V_\rho(A)\ge V_\rho^{\rm cl}(A)+c_{\rm opt}(\rho)Q_\rho(A),
\end{equation}
with the same optimal coefficient $c_{\rm opt}(\rho)$ as in the corresponding multiplicative bound without the classical contribution.
Thus, the maximal classical contribution can be added without weakening the optimal noncommutative term.
This structure also makes explicit how the variance separates into contributions associated with the diagonal and off-diagonal parts of the observable relative to the state.

We further obtained a complete characterization of universal exactness.
Among full-rank nonmaximally mixed states, each of the sharpened power-commutator, Wigner--Yanase--Dyson, and quantum-Fisher-information relations is an identity for every observable if and only if the state has exactly two distinct eigenvalues.
Thus, the two-eigenvalue condition is not merely sufficient, but precisely characterizes universal exactness within this class of states.
The origin of this result lies in the eigenvalue structure.
When the state has only two distinct eigenvalues, every noncommuting contribution connects the same pair of eigenspaces and is therefore governed by the same extremal pair of eigenvalues.
By contrast, the presence of a third distinct eigenvalue allows observables to couple eigenspaces associated with a nonextremal pair, for which the underlying scalar estimate is strict.
Qubit systems provide the simplest realization of this structure, and, including the boundary cases, the sharpened relations become exact identities for every qubit state.

Finally, applying the single-observable relations separately to two observables yields corresponding product-form uncertainty relations.
For qubit systems, the resulting optimized noncommutative bound is no weaker than the Robertson bound for every pair of spin observables.
For mixed qubit states with $1/2<P<1$, its uniform average also exceeds the corresponding averaged Schr\"odinger and Luo-type bounds.

Taken together, these results provide a perspective on quantum uncertainty that complements the conventional trade-off approach.
They show that the uncertainty of a single observable already contains a rich structure determined by the state, its spectrum, and its incompatibility with that observable.
In particular, the extremal eigenvalues determine the strongest universal coefficient at fixed state, while the maximal compatible part of the observable supplies a classical contribution that can be retained without weakening this optimal quantum term.
For full-rank states with exactly two distinct eigenvalues, these ingredients are sufficient to determine the variance exactly for every observable.

\section*{Author contributions}
All authors contributed to the development of the work and to the preparation of the manuscript.
The manuscript was mainly written by H.Y. and G.K.
ChatGPT (OpenAI) was used during manuscript preparation for language editing, stylistic restructuring, consistency checks, drafting assistance, and limited bibliographic research.
The authors remain fully responsible for the mathematical derivations, scientific conclusions, and final manuscript.

\section*{Acknowledgments}

We thank Satoya Imai for bringing Ref.~\cite{TothPetz2013Extremal} to our attention.
G.K. was supported by JSPS KAKENHI Grant No. 24K06873.

\appendix
\section{Proof of Lemmas}\label{app:lemma}

Lemmas~\ref{lem:sq}, \ref{lem:wyd-scalar}, and~\ref{lem:qfi-scalar}
are all special cases of a single scalar inequality.

\begin{lemma}\label{lem:scalar-ieq}
Let \(0\le m<M\), let \(p,q,\gamma>0\), and assume that
\begin{equation}\label{eq:gen-condition}
p+q\ge\gamma.
\end{equation}
Then, for all \(x,y\in[m,M]\),
\begin{equation}\label{eq:general-scalar}
(x+y)^\gamma
\ge
\frac{(M+m)^\gamma}
{(M^p-m^p)(M^q-m^q)}
(x^p-y^p)(x^q-y^q).
\end{equation}
Moreover, if \(m>0\), \(x\neq y\), and
\(\{x,y\}\neq\{m,M\}\), then the inequality is strict.
\end{lemma}

\begin{proof}
We first prove the result for \(m>0\).

Because both sides of Eq.~\eqref{eq:general-scalar} are invariant under
the exchange \(x\leftrightarrow y\), it is enough to consider
\begin{equation}\label{eq:gen-order}
x\ge y.
\end{equation}

If \(x=y\), then
\begin{equation}
(x^p-y^p)(x^q-y^q)=0,
\end{equation}
whereas
\[
(x+y)^\gamma\ge0.
\]
Hence Eq.~\eqref{eq:general-scalar} is immediate.
We therefore assume from now on that
\begin{equation}\label{eq:gen-strict-order}
x>y>0.
\end{equation}

Introduce the dimensionless ratios
\begin{equation}\label{eq:gen-ratios}
t:=\frac{x}{y},
\qquad
R:=\frac{M}{m}.
\end{equation}
Since
\[
m\le y<x\le M,
\]
we have
\begin{equation}\label{eq:gen-ratio-range}
1<t\le R.
\end{equation}
Also set
\begin{equation}
a:=p+q.
\end{equation}
By assumption~\eqref{eq:gen-condition},
\begin{equation}\label{eq:gen-a-gamma}
a-\gamma\ge0.
\end{equation}

We now rewrite the ratio appearing in
Eq.~\eqref{eq:general-scalar}.
Since \(t=x/y\), we have \(y=x/t\), and therefore
\begin{align}
x^p-y^p
&=
x^p-\left(\frac{x}{t}\right)^p
=
x^p\left(1-t^{-p}\right),
\\
x^q-y^q
&=
x^q-\left(\frac{x}{t}\right)^q
=
x^q\left(1-t^{-q}\right),
\\
x+y
&=
x+\frac{x}{t}
=
x\left(1+t^{-1}\right)
=
x\frac{1+t}{t}.
\end{align}
Substituting these expressions gives
\begin{align}
\frac{(x^p-y^p)(x^q-y^q)}
     {(x+y)^\gamma}
&=
\frac{
x^{p+q}
\left(1-t^{-p}\right)
\left(1-t^{-q}\right)
}{
x^\gamma
\left(\frac{1+t}{t}\right)^\gamma
}
\nonumber\\
&=
x^{a-\gamma}
\left(\frac{t}{1+t}\right)^\gamma
\left(1-t^{-p}\right)
\left(1-t^{-q}\right).
\label{eq:gen-rewrite}
\end{align}

We estimate the two parts on the right-hand side of
Eq.~\eqref{eq:gen-rewrite} separately.

First, since \(a-\gamma\ge0\) by
Eq.~\eqref{eq:gen-a-gamma}, the function
\(u\mapsto u^{a-\gamma}\) is nondecreasing on \((0,\infty)\).
Together with \(x\le M\), this gives
\begin{equation}\label{eq:gen-x-bound}
x^{a-\gamma}\le M^{a-\gamma}.
\end{equation}
Notice that when \(a=\gamma\), both sides of
Eq.~\eqref{eq:gen-x-bound} are equal to \(1\).
When \(a>\gamma\), the inequality is strict unless \(x=M\).

Next, define
\begin{equation}\label{eq:gen-G-def}
G(t):=
\left(\frac{t}{1+t}\right)^\gamma
\left(1-t^{-p}\right)
\left(1-t^{-q}\right),
\qquad t>1.
\end{equation}
We show explicitly that \(G\) is strictly increasing on \((1,\infty)\).

For \(t>1\),
\begin{align}
\frac{d}{dt}\left(\frac{t}{1+t}\right)
&=
\frac{1}{(1+t)^2}
>0,
\\
\frac{d}{dt}\left(1-t^{-p}\right)
&=
p\,t^{-p-1}
>0,
\\
\frac{d}{dt}\left(1-t^{-q}\right)
&=
q\,t^{-q-1}
>0.
\end{align}
Hence
\[
\frac{t}{1+t},
\qquad
1-t^{-p},
\qquad
1-t^{-q}
\]
are all strictly increasing on \(t>1\).
They are also strictly positive there.
Since \(\gamma>0\),
\[
\left(\frac{t}{1+t}\right)^\gamma
\]
is likewise strictly positive and strictly increasing.

Since \(1<t\le R\), we consequently obtain
\begin{align}
&\left(\frac{t}{1+t}\right)^\gamma
\left(1-t^{-p}\right)
\left(1-t^{-q}\right)
\nonumber\\
&\qquad\le
\left(\frac{R}{1+R}\right)^\gamma
\left(1-R^{-p}\right)
\left(1-R^{-q}\right).
\label{eq:gen-t-bound}
\end{align}
Moreover, because \(G\) is strictly increasing, the inequality in
Eq.~\eqref{eq:gen-t-bound} is strict whenever
\begin{equation}\label{eq:gen-strict-t}
1<t<R.
\end{equation}

Combining Eqs.~\eqref{eq:gen-rewrite},
\eqref{eq:gen-x-bound}, and~\eqref{eq:gen-t-bound}, we obtain
\begin{align}
\frac{(x^p-y^p)(x^q-y^q)}
     {(x+y)^\gamma}
&\le
M^{a-\gamma}
\left(\frac{R}{1+R}\right)^\gamma
\left(1-R^{-p}\right)
\left(1-R^{-q}\right).
\label{eq:gen-before-endpoint}
\end{align}

It remains to rewrite the right-hand side entirely in terms of \(m\)
and \(M\).
Using \(R=M/m\) and \(a=p+q\), we have
\begin{align}
\frac{R}{1+R}
&=
\frac{M/m}{1+M/m}
=
\frac{M}{M+m},
\\
1-R^{-p}
&=
1-\left(\frac{m}{M}\right)^p
=
\frac{M^p-m^p}{M^p},
\\
1-R^{-q}
&=
1-\left(\frac{m}{M}\right)^q
=
\frac{M^q-m^q}{M^q}.
\end{align}
Hence
\begin{align}
&M^{a-\gamma}
\left(\frac{R}{1+R}\right)^\gamma
\left(1-R^{-p}\right)
\left(1-R^{-q}\right)
\nonumber\\
&\quad=
M^{p+q-\gamma}
\left(\frac{M}{M+m}\right)^\gamma
\frac{M^p-m^p}{M^p}
\frac{M^q-m^q}{M^q}
\nonumber\\
&\quad=
M^{p+q-\gamma}
\frac{M^\gamma}{(M+m)^\gamma}
\frac{(M^p-m^p)(M^q-m^q)}{M^{p+q}}
\nonumber\\
&\quad=
\frac{(M^p-m^p)(M^q-m^q)}
     {(M+m)^\gamma}.
\label{eq:gen-endpoint-rewrite}
\end{align}
In the last equality, the powers of \(M\) cancel because
\[
M^{p+q-\gamma}M^\gamma M^{-(p+q)}=1.
\]

Substituting Eq.~\eqref{eq:gen-endpoint-rewrite} into
Eq.~\eqref{eq:gen-before-endpoint} yields
\begin{equation}\label{eq:gen-ratio-final}
\frac{(x^p-y^p)(x^q-y^q)}
     {(x+y)^\gamma}
\le
\frac{(M^p-m^p)(M^q-m^q)}
     {(M+m)^\gamma}.
\end{equation}
Because \(x+y>0\), multiplication by \((x+y)^\gamma\) gives
\[
(x^p-y^p)(x^q-y^q)
\le
\frac{(M^p-m^p)(M^q-m^q)}
     {(M+m)^\gamma}
(x+y)^\gamma.
\]
Since \(M>m\ge0\) and \(p,q>0\),
\[
M^p-m^p>0,
\qquad
M^q-m^q>0.
\]
Therefore we may divide by their product and rearrange to obtain
precisely
\[
(x+y)^\gamma
\ge
\frac{(M+m)^\gamma}
     {(M^p-m^p)(M^q-m^q)}
(x^p-y^p)(x^q-y^q).
\]
This proves Eq.~\eqref{eq:general-scalar} for \(m>0\).

We now record the strictness property, since it will be used in the
proof of Theorem~\ref{thm:universal-exactness}.
Suppose \(m>0\) and \(x>y\).
If \(t=R\), then
\[
\frac{x}{y}=\frac{M}{m},
\]
or equivalently,
\[
xm=My.
\]
Since \(x\le M\), this implies
\[
My=xm\le Mm,
\]
and hence \(y\le m\).
But \(y\ge m\), so necessarily
\[
y=m.
\]
Substituting this back into \(x/y=M/m\) gives
\[
x=M.
\]
Thus
\begin{equation}\label{eq:tR-endpoints}
t=R
\quad\Longleftrightarrow\quad
(x,y)=(M,m)
\end{equation}
under the ordering \(x>y\).

Consequently, if \(x>y\) and
\[
(x,y)\neq(M,m),
\]
then necessarily \(1<t<R\).
Since \(G\) is strictly increasing, Eq.~\eqref{eq:gen-t-bound} is then
strict. It follows that Eq.~\eqref{eq:gen-ratio-final}, and hence
Eq.~\eqref{eq:general-scalar}, is also strict.
By symmetry under \(x\leftrightarrow y\), the same statement holds
when \(y>x\).
Therefore, for \(m>0\) and \(x\neq y\), equality can occur only when
\[
\{x,y\}=\{m,M\}.
\]
In particular, every distinct pair other than \(\{m,M\}\) gives a
strict inequality.

We finally treat the case \(m=0\).

Assume first that \(x>0\) and \(y>0\).
Choose any \(m'\) satisfying
\begin{equation}
0<m'\le\min\{x,y\}.
\end{equation}
Then
\[
x,y\in[m',M],
\]
so the result already proved for a positive lower endpoint applies:
\begin{equation}\label{eq:gen-mprime}
(x+y)^\gamma
\ge
\frac{(M+m')^\gamma}
     {(M^p-(m')^p)(M^q-(m')^q)}
(x^p-y^p)(x^q-y^q).
\end{equation}
Now let \(m'\to0^+\).
Since \(p,q,\gamma>0\),
\begin{align}
(M+m')^\gamma
&\longrightarrow M^\gamma,
\\
M^p-(m')^p
&\longrightarrow M^p,
\\
M^q-(m')^q
&\longrightarrow M^q.
\end{align}
The denominator therefore converges to
\[
M^{p+q}>0.
\]
Hence, by continuity, Eq.~\eqref{eq:gen-mprime} gives
\begin{equation}
(x+y)^\gamma
\ge
\frac{M^\gamma}{M^pM^q}
(x^p-y^p)(x^q-y^q),
\end{equation}
which is exactly Eq.~\eqref{eq:general-scalar} with \(m=0\).

If \(x=0\) or \(y=0\), the same conclusion follows by taking a
sequence of positive values converging to the vanishing variable and
using continuity of both sides of Eq.~\eqref{eq:general-scalar}.
Thus Eq.~\eqref{eq:general-scalar} holds for all \(0\le m<M\).
\end{proof}

We now recover the three scalar lemmas used in the main text.

First, set
\begin{equation}
p=q=s,
\qquad
\gamma=1,
\end{equation}
with \(s\ge1/2\).
Then
\[
p+q=2s\ge1=\gamma,
\]
so all assumptions of Lemma~\ref{lem:scalar-ieq} are satisfied.
Equation~\eqref{eq:general-scalar} therefore gives
\begin{equation}
x+y
\ge
\frac{M+m}{(M^s-m^s)^2}
(x^s-y^s)^2,
\end{equation}
which is Lemma~\ref{lem:sq}.

Second, set
\begin{equation}
p=\alpha,
\qquad
q=1-\alpha,
\qquad
\gamma=1,
\end{equation}
with \(0<\alpha<1\).
Then \(p,q>0\) and
\[
p+q
=
\alpha+(1-\alpha)
=
1
=
\gamma.
\]
Hence Lemma~\ref{lem:scalar-ieq} gives
\begin{equation}
x+y
\ge
\frac{M+m}
{(M^\alpha-m^\alpha)
 (M^{1-\alpha}-m^{1-\alpha})}
(x^\alpha-y^\alpha)
(x^{1-\alpha}-y^{1-\alpha}),
\end{equation}
which is Lemma~\ref{lem:wyd-scalar}.

Third, set
\begin{equation}
p=q=1,
\qquad
\gamma=2.
\end{equation}
Then
\[
p+q=2=\gamma,
\]
and Lemma~\ref{lem:scalar-ieq} gives
\begin{equation}
(x+y)^2
\ge
\frac{(M+m)^2}{(M-m)^2}
(x-y)^2.
\label{eq:qfi-from-general}
\end{equation}
Whenever \(x+y>0\), division by \(x+y\) yields
\begin{equation}
x+y
\ge
\left(\frac{M+m}{M-m}\right)^2
\frac{(x-y)^2}{x+y},
\end{equation}
which is Lemma~\ref{lem:qfi-scalar}.

In this way, the scalar estimates underlying Theorems~\ref{thm:s-main}, \ref{thm:wyd-extension}, and~\ref{thm:QFI-extension} are all obtained from the same inequality.
The strictness statement in Lemma~\ref{lem:scalar-ieq} further shows that, when \(m>0\), every distinct pair of points in \([m,M]\) other than \(\{m,M\}\) gives a strict inequality.
This fact will be used in the proof of the necessity part of
Theorem~\ref{thm:universal-exactness}.

\section{Uniform averages on the sphere and averaged lower bounds}
\label{app:qubit-average}

In this appendix, we derive the averaged lower bounds used in
Sec.~\ref{sec:product}.  We first recall a formula for the uniform average on
the unit sphere.

Let $d\mu$ be the normalized uniform measure on
$S^{d-1}\subset\mathbb{R}^d$.  For a function $f$ on $S^{d-1}$, we
write
\begin{equation}
  \langle f(\bm x)\rangle_{\rm av}
  :=
  \int_{S^{d-1}} d\mu\, f(\bm x).
  \label{eq:app-spherical-average}
\end{equation}
For $\bm x=(x_1,\ldots,x_d)\in S^{d-1}$, we have
\begin{equation}
  \langle x_jx_k\rangle_{\rm av}
  =
  \int_{S^{d-1}} d\mu\, x_jx_k
  =
  \frac{\delta_{jk}}{d}.
  \label{eq:app-spherical-second-moment}
\end{equation}
By rotational
symmetry, $\int d\mu\, x_i^2$ is independent of $i$, and
$\int d\mu\, x_i x_j$ is independent of the pair $i\ne j$.  Put
\begin{equation}
  c:=\int_{S^{d-1}} d\mu\, x_i^2,
  \qquad
  c':=\int_{S^{d-1}} d\mu\, x_i x_j \quad (i\ne j).\nonumber
\end{equation}
Since $\sum_i x_i^2=1$ on $S^{d-1}$, we have
\begin{equation}
  1
  =
  \int_{S^{d-1}} d\mu\, \sum_{i=1}^d x_i^2
  =
  \sum_{i=1}^d \int_{S^{d-1}} d\mu\, x_i^2
  =
  dc.\nonumber
\end{equation}
Thus $c=1/d$.  To determine $c'$, observe that
\begin{align}
  \int_{S^{d-1}} d\mu\,
  \left(\sum_{i=1}^d x_i\right)^2
  &=
  \sum_{i=1}^d \int_{S^{d-1}} d\mu\, x_i^2
  +
  \sum_{i\ne j}\int_{S^{d-1}} d\mu\, x_i x_j       \nonumber\\
  &=
  1+d(d-1)c'.
\end{align}
On the other hand, by choosing a new orthonormal coordinate system
$\bm x'=(x'_1,\ldots,x'_d)$ such that
\begin{equation}
  x'_1=\frac{1}{\sqrt d}\sum_{i=1}^d x_i,
\end{equation}
we obtain
\begin{equation}
  \int_{S^{d-1}} d\mu\,
  \left(\sum_{i=1}^d x_i\right)^2
  =
  \int_{S^{d-1}} d\mu\, d(x'_1)^2
  =
  d\frac{1}{d}
  =
  1.
\end{equation}
Hence $c'=0$, and Eq.~\eqref{eq:app-spherical-second-moment}
follows.

We now apply Eq.~\eqref{eq:app-spherical-second-moment} to the
averaging of qubit uncertainty bounds.  Let
\begin{equation}
  A=\bm a\cdot\bm\sigma,\qquad
  B=\bm b\cdot\bm\sigma,
  \qquad
  |\bm a|=|\bm b|=1,\nonumber
\end{equation}
where $\bm a,\bm b\in S^2$ are independently and uniformly
distributed.  By rotational invariance, we may take the Bloch vector
of the state to be in the $z$-direction:
\begin{equation}
  \rho=\frac{I+r\sigma_z}{2},
  \qquad
  0\le r\le1 .
\end{equation}
The purity is
\begin{equation}\label{eq:Bloch-purity}
  P:=\Tr\rho^2
  =
  \frac{1+r^2}{2},
  \qquad
  r^2=2P-1 .
\end{equation}
In what follows, the average is taken independently over
$\bm a$ and $\bm b$.  From Eq.~\eqref{eq:app-spherical-second-moment}
with $d=3$, we use
\begin{equation}
  \langle a_i\rangle_{\rm av}=0,
  \qquad
  \langle a_ia_j\rangle_{\rm av}
  =
  \frac{\delta_{ij}}{3},
\end{equation}
and the same relations for $\bm b$.

Before computing the averaged lower bounds, let us first evaluate the
averaged variance product itself.  Since Eq.~\eqref{eq:mainURQubit} is an
identity in qubit systems, the sharper product bound \eqref{Our1Sharp}
coincides with $V_\rho(A)V_\rho(B)$.  Thus, its averaged value is obtained
by averaging the variance product itself.

Since $A=\bm a\cdot\bm\sigma$ and $B=\bm b\cdot\bm\sigma$ are normalized observables, we have $A^2=B^2=I$.  Hence
\begin{equation}
  V_\rho(A)
  =
  \langle A^2\rangle_\rho-\langle A\rangle_\rho^2
  =
  1-r^2a_3^2,
  \qquad
  V_\rho(B)
  =
  1-r^2b_3^2 .
\end{equation}
Therefore, using the independence of $\bm a$ and $\bm b$, we obtain
\begin{align}
  \left\langle
  V_\rho(A)V_\rho(B)
  \right\rangle_{\rm av}
  &=
  \left\langle
  (1-r^2a_3^2)(1-r^2b_3^2)
  \right\rangle_{\rm av}                                  \nonumber\\
  &=1-r^2\langle a_3^2\rangle_{\rm av}
  -r^2\langle b_3^2\rangle_{\rm av}
  +r^4
  \langle a_3^2\rangle_{\rm av}
  \langle b_3^2\rangle_{\rm av}                           \nonumber\\
  &=1-\frac{2r^2}{3}+\frac{r^4}{9}                           \nonumber\\
  &=\left(1-\frac{r^2}{3}\right)^2  =\frac{(3-r^2)^2}{9}.
\end{align}
By~\eqref{eq:Bloch-purity}, this can be written as
\begin{equation}
  \left\langle
  V_\rho(A)V_\rho(B)
  \right\rangle_{\rm av}
  =
  \frac{(3-(2P-1))^2}{9}
  =
  \frac{4(2-P)^2}{9}.
  \label{eq:app-avg-variance-product}
\end{equation}
In particular, this averaged variance product decreases from $1$ for the maximally mixed state $P=1/2$ to $4/9$ for pure states $P=1$. 
Since our bound \eqref{Our1Sharp} coincides with the variance in the qubit case, the average $\langle B_{\rho}^{\rm sharp}\rangle_{\rm av}$ in \eqref{eq:avg-ours-sharp} also coincides with \eqref{eq:app-avg-variance-product}.

In what follows, we compute the averaged bounds for the Robertson, Schr\"odinger, Luo, and our bound \eqref{Our1}, in this order. 

We start with the Robertson bound.
Since
\begin{equation}
  [\bm a\cdot\bm\sigma,\bm b\cdot\bm\sigma]
  =
  2i(\bm a\times\bm b)\cdot\bm\sigma,\nonumber
\end{equation}
we have
\begin{equation}
  \langle[A,B]\rangle_\rho
  =
  2ir\,(\bm a\times\bm b)_3
  =
  2ir(a_1b_2-a_2b_1).
\end{equation}
Therefore
\begin{equation}
  B_{\rm R}
  :=
  \frac{1}{4}|\langle[A,B]\rangle_\rho|^2
  =
  r^2(a_1b_2-a_2b_1)^2 .
\end{equation}
Taking the uniform average, we obtain
\begin{align}
  \langle B_{\rm R}\rangle_{\rm av}
  &=
  r^2
  \left\langle
  (a_1b_2-a_2b_1)^2
  \right\rangle_{\rm av}                                      \nonumber\\
  &=
  r^2
  \left(
    \langle a_1^2\rangle_{\rm av}
    \langle b_2^2\rangle_{\rm av}
    +
    \langle a_2^2\rangle_{\rm av}
    \langle b_1^2\rangle_{\rm av}
  \right)                                                       \nonumber\\
  &=
  \frac{2}{9}r^2
  =
  \frac{2}{9}(2P-1).
  \label{eq:app-avg-robertson}
\end{align}

Next, we compute the Schr\"odinger bound.  We have
\begin{equation}
  \langle A\rangle_\rho=ra_3,
  \qquad
  \langle B\rangle_\rho=rb_3,\nonumber
\end{equation}
and
\begin{equation}
  \frac{1}{2}\langle\{A,B\}\rangle_\rho
  =
  \bm a\cdot\bm b .
\end{equation}
Hence
\begin{align}
  {\rm Cov}_\rho(A,B)
  &=
  \frac{1}{2}\langle\{A,B\}\rangle_\rho
  -
  \langle A\rangle_\rho\langle B\rangle_\rho       \nonumber\\
  &=
  a_1b_1+a_2b_2+(1-r^2)a_3b_3 .
\end{align}
Using Eq.~\eqref{eq:app-spherical-second-moment}, we find
\begin{align}
  \left\langle
  {\rm Cov}_\rho(A,B)^2
  \right\rangle_{\rm av}
  &=
  \frac{1}{9}
  +
  \frac{1}{9}
  +
  \frac{(1-r^2)^2}{9}                                      \nonumber\\
  &=
  \frac{2+(1-r^2)^2}{9}.
\end{align}
Combining this with Eq.~\eqref{eq:app-avg-robertson}, the averaged
Schr\"odinger bound is
\begin{align}
  \langle B_{\rm S}\rangle_{\rm av}
  &:=
  \left\langle
  \frac{1}{4}|\langle[A,B]\rangle_\rho|^2
  +
  {\rm Cov}_\rho(A,B)^2
  \right\rangle_{\rm av}                                  \nonumber\\
  &=
  \frac{2r^2}{9}
  +
  \frac{2+(1-r^2)^2}{9}                                   \nonumber\\
  &=
  \frac{3+r^4}{9}  =
  \frac{4}{9}(P^2-P+1).
  \label{eq:app-avg-schrodinger}
\end{align}

We turn to the Luo bound.  The eigenvalues of $\rho$ are
\begin{equation}
  \lambda_\pm=\frac{1\pm r}{2}.
\end{equation}
In the eigenbasis of $\rho$, the off-diagonal matrix elements of
$A$ and $B$ satisfy
\begin{equation}
  |A_{12}|^2=a_1^2+a_2^2=1-a_3^2,
  \qquad
  |B_{12}|^2=b_1^2+b_2^2=1-b_3^2.
\end{equation}
Moreover,
\begin{equation}
  \|[A,\sqrt{\rho}]\|^2
  =
  2(\sqrt{\lambda_+}-\sqrt{\lambda_-})^2|A_{12}|^2 .
\end{equation}
Therefore Luo's single-observable estimate gives the product bound
\begin{align}
  B_{\rm L}
  &:=
  \frac{1}{4}
  \|[A,\sqrt{\rho}]\|^2
  \|[B,\sqrt{\rho}]\|^2                              \nonumber\\
  &=
  (\sqrt{\lambda_+}-\sqrt{\lambda_-})^4
  |A_{12}|^2|B_{12}|^2 .
\end{align}
Since
\begin{equation}
  (\sqrt{\lambda_+}-\sqrt{\lambda_-})^2
  =
  1-\sqrt{1-r^2},
\end{equation}
we obtain
\begin{align}
  \langle B_{\rm L}\rangle_{\rm av}
  &=
  (1-\sqrt{1-r^2})^2
  \left\langle
  (1-a_3^2)(1-b_3^2)
  \right\rangle_{\rm av}                                \nonumber\\
  &=
  (1-\sqrt{1-r^2})^2
  \left(1-\frac{1}{3}\right)^2                          \nonumber\\
  &=
  \frac{4}{9}(1-\sqrt{1-r^2})^2                         \nonumber\\
  &=
  \frac{4}{9}\left(1-\sqrt{2(1-P)}\right)^2 .
  \label{eq:app-avg-luo}
\end{align}

Finally, we compute the averaged values of the optimal bound \eqref{Our1}. 
Noting that
\begin{equation}
  \|[A,\rho^s]\|^2
  =
  2(\lambda_+^s-\lambda_-^s)^2|A_{12}|^2,
\end{equation}
the product lower bound reads
\begin{align}
  B_{\rho}^{\rm opt}
  &:=
  \frac{1}{4}
  \frac{(\lambda_++\lambda_-)^2}
       {(\lambda_+^s-\lambda_-^s)^4}
  \|[A,\rho^s]\|^2
  \|[B,\rho^s]\|^2                              \nonumber\\
  &=
  (\lambda_++\lambda_-)^2
  |A_{12}|^2|B_{12}|^2  \nonumber \\
  &=
  |A_{12}|^2|B_{12}|^2 = (1-a_3^2)(1-b_3^2),
\end{align}
which is independent of the parameter $s$.
Taking the average, we obtain
\begin{align}
  \langle B_{\rho}^{\rm opt}\rangle_{\rm av}
  &=
  \left\langle
  (1-a_3^2)(1-b_3^2)
  \right\rangle_{\rm av}  \nonumber\\
  &=
  \left(1-\frac{1}{3}\right)^2  =\frac{4}{9}.
  \label{eq:app-avg-optimal}
\end{align}

Equations~\eqref{eq:app-avg-variance-product}, \eqref{eq:app-avg-robertson},
\eqref{eq:app-avg-schrodinger}, \eqref{eq:app-avg-luo}, and
\eqref{eq:app-avg-optimal} give, respectively, the averaged lower
bounds quoted in Sec.~\ref{sec:product}, namely
Eqs.~\eqref{eq:avg-ours-sharp}, \eqref{eq:avg-R}, \eqref{eq:avg-S}, \eqref{eq:avg-Luo}, and
\eqref{eq:avg-ours}.

\bibliographystyle{quantum}
\bibliography{ref_UR}

\end{document}